\documentclass[12pt,a4paper]{article}

\usepackage[margin=1in]{geometry}
\usepackage[doublespacing]{setspace}

\usepackage{amsmath}
\usepackage{amssymb}
\usepackage{amsthm}
\usepackage{mathtools}
\usepackage{mathrsfs}
\usepackage{bm}

\usepackage{graphicx}
\usepackage{booktabs}
\usepackage[ruled,vlined]{algorithm2e}
\usepackage{multirow}
\usepackage{subcaption}
\usepackage{rotating}

\usepackage{enumitem}

\usepackage{xcolor}
\usepackage{url}
\usepackage[hidelinks]{hyperref}

\usepackage[authoryear,round]{natbib}

\usepackage{appendix}

\graphicspath{{Fig/}}

\theoremstyle{plain}
\newtheorem{theorem}{Theorem}
\newtheorem{proposition}[theorem]{Proposition}

\newtheorem{corollary}{Corollary}

\theoremstyle{definition}
\newtheorem{definition}{Definition}
\newtheorem{assumption}{Assumption}

\theoremstyle{remark}
\newtheorem{remark}{Remark}

\newcommand{\Pall}{\mathcal{P}_{\mathrm{all}}}
\newcommand{\Rterm}{\mathcal{R}_T}
\newcommand{\Rstar}{\mathcal{R}_T^{\star}}

\title{\textbf{BRACE: Differential Privacy for Dense Associative Memory with LSR Energy}}

\author{
Chang Qu\\[0.5em]
Department of Mathematics and Statistics, University of Ottawa\\
\and
Zhaoyang Shi\thanks{Corresponding author:
\href{mailto:zyshi10m@outlook.com}{zyshi10m@outlook.com}}\\[0.5em]
Center for Applied Mathematics, Fudan University
}
\date{}

\begin{document}

\maketitle

%\editor{Associate Editor: Name}

%\abstract{
%\textbf{Motivation:} .\\
%\textbf{Results:} .\\
%\textbf{Availability:} .\\
%\textbf{Contact:} \href{name@email.com}{name@email.com}\\
%\textbf{Supplementary information:} Supplementary data are available at \textit{Journal Name}
%online.}

\begin{abstract}
Dense associative memory (DAM) provides an energy-based framework for memory retrieval with close connections to attention mechanisms in modern artificial intelligence. Despite growing interest in differential privacy for AI, the privacy of DAM retrieval dynamics remains relatively unexplored. In this paper, we develop a differential privacy framework for log-sum-ReLU (LSR) dense associative memory, whose finite-support retrieval dynamics pose distinctive challenges for privacy-preserving computation. We propose the Boundary-Responsive Adaptive Correction Evolution (BRACE) algorithm, a differentially private retrieval mechanism for LSR-DAM that adaptively corrects boundary-sensitive perturbations to control their cumulative effect over the retrieval trajectory. In theory, we prove that our method is minimax optimal by deriving dimension-independent terminal and full-trajectory retrieval error rates, with optimal dependence on the inverse temperature and, in the growing-horizon regime, the retrieval horizon. We further establish central limit theorems that enable uncertainty quantification for private retrieval by characterizing its asymptotic distribution and the additional variability introduced by privacy. Numerical experiments compare our proposed method with baseline differential privacy approaches and evaluate its retrieval accuracy. Together, our results provide a theoretical foundation for optimal privacy-preserving retrieval and uncertainty quantification in energy-based associative memory systems.
\end{abstract}

% \boxedtext{
% \begin{itemize}
% \item Key boxed text here.
% \item Key boxed text here.
% \item Key boxed text here.
% \end{itemize}}

\section{Introduction}
\label{sec:introduction}
% Introduction and related-work discussion to be written by the author.
Differential privacy (DP) has become a fundamental framework for enabling
privacy-preserving machine learning
\citep{dwork2006calibrating,dwork2014algorithmic,abadi2016deep,ponomareva2023dp},
particularly in modern artificial intelligence systems that increasingly rely
on sensitive and personalized data
\citep{yu2021differentially,lin2024differentially,sander2024differentially}.
With the emergence of memory-augmented AI systems, such as MemGPT
\citep{packer2023memgpt}, LongMem \citep{wang2023augmenting}, and other
memory-enhanced LLM architectures \citep{schuurmans2023memory}, protecting
stored information and retrieval processes becomes increasingly important.
Recently, dense associative memory (DAM) has emerged as an energy-based
framework for memory storage and retrieval, with close connections to modern
AI systems, including attention mechanisms and memory-augmented architectures
\citep{ramsauer2021hopfield,krotov2021large,hoover2023energy}.
However, despite extensive progress on DP for machine learning models, the
privacy of DAM retrieval dynamics remains largely unexplored. This challenge
is particularly relevant in applications involving sensitive memories, such as
medical record retrieval \citep{wang2025activating}, personalized LLM memory
modules \citep{packer2023memgpt}, and user-specific AI assistants
\citep{yuan2023personalized}. Therefore, developing principled DP mechanisms
for DAM retrieval is essential for trustworthy memory-based AI systems.

Among existing DAM formulations, the log-sum-exponential (LSE) energy has been
the dominant choice \citep{ramsauer2021hopfield,krotov2016dense} due to its
close connection with Transformer attention. LSE-DAM performs retrieval through
smooth similarity-based weighting of memories, but faces an inherent trade-off
between memorization and generalization. Recently, the log-sum-ReLU (LSR)
energy was introduced as a compact-support alternative to LSE-DAM to improve
this trade-off \citep{hoover2026dense,santos2026sparse}. By
replacing exponential interactions with compactly supported interactions,
LSR-DAM enables localized retrieval and distinct memory properties. However,
this finite-support mechanism introduces a fundamental challenge for DP.
Unlike LSE-DAM, where retrieval weights vary smoothly, LSR-DAM relies on a hard
retrieval boundary that determines the active memory set. Privacy perturbations
can therefore trigger boundary switching, causing discontinuous changes in the
retrieval operator and errors that accumulate over iterative retrieval. As a
result, existing DP approaches based on smooth perturbation analysis cannot be
directly applied. This motivates the following fundamental question:

\textit{How can we design a differentially private retrieval mechanism for
LSR-DAM that adaptively corrects boundary-induced instability while achieving
accurate and theoretically optimal retrieval?}

In this paper, we answer this question by proposing \textbf{BRACE}
(\textbf{B}oundary-\textbf{R}esponsive \textbf{A}daptive
\textbf{C}orrection \textbf{E}volution), the first DP framework specifically
designed for LSR-DAM. BRACE separates privacy perturbations from retrieval
instability by identifying memory points whose membership changes under noisy
retrieval and applying an adaptive correction before injecting calibrated DP
noise. This boundary-aware design enables stable private retrieval while
preserving the compact-support dynamics of LSR-DAM.

Our contributions are summarized as follows:
\begin{enumerate}
    \item We formulate the first DP framework for LSR-DAM and propose BRACE,
    which resolves the boundary instability uniquely induced by compact-support
    associative memory retrieval.
    \item We establish rigorous theoretical guarantees for BRACE, including
    differential privacy and minimax-optimal retrieval error bounds, and
    demonstrate its superiority over commonly used DP baselines through
    experiments.
    \item We develop a CLT-based uncertainty quantification framework for
    private LSR-DAM retrieval, providing principled characterization of
    privacy-induced uncertainty in energy-based AI systems.
\end{enumerate}

By establishing differential privacy for LSR-DAM, our work not only advances
privacy-preserving associative memory systems, but also enables secure
memory-augmented AI architectures for applications involving sensitive and
personalized information.

\section{Preliminaries and Problem Setup}
Let $n,d\in\mathbb N$ and write $[n]=\{1,\ldots,n\}$.
For a sequence $(Y_s)$, write
\[
Y_{<t}=(Y_0,\ldots,Y_{t-1}),
\]
and use $y_{<t}$ for a fixed realization of this history.
We write $\mathbf 1\{\cdot\}$ for the indicator function and $I_k$ for
the $k\times k$ identity matrix.
Let $D=(X_1,\ldots,X_n)\in(\mathbb R^d)^n$ be the private memory bank.
The initial query $q\in\mathbb R^d$, retrieval radius $r>0$, and horizon
$T\ge1$ are public. Throughout, $\|\cdot\|$ is the Euclidean norm.

Dense associative memory (DAM) retrieves patterns by iterating an
energy-based update. For inverse temperature $\beta>0$, the
log-sum-exponential (LSE) energy and update are
\[
\begin{aligned}
E_{\mathrm{LSE}}(z;D)
&=
-\frac{1}{\beta}
\log\left(
\sum_{i=1}^n e^{\beta z^\top X_i}
\right),\\
z_{t+1}
&=
\frac{\sum_{i=1}^n e^{\beta z_t^\top X_i}X_i}
{\sum_{i=1}^n e^{\beta z_t^\top X_i}}.
\end{aligned}
\]
The log-sum-ReLU (LSR) energy replaces exponential interactions with
compactly supported interactions:
\[
\begin{aligned}
E_{\mathrm{LSR}}(z;D) &= -\log\!\left(\sum_{i=1}^n (r^2-\|X_i-z\|^2)_+\right),\\
(a)_+ &= \max\{a,0\}, \qquad r := \sqrt{\frac{2}{\beta}}.
\end{aligned}
\]
{Unlike the global LSE update, LSR retrieval averages only memories within radius \(r\). This compact support yields bounded movement and finite fixed-center sensitivity without clipping or globally bounded memory vectors, and the sensitivity decreases with the local active count. The trade-off is boundary discontinuity, which motivates BRACE.}
\begin{definition}[Central differential privacy \citep{dwork2014algorithmic}]
\label{def:central-dp}
A mechanism $M(D)=(Y_1,\ldots,Y_T)\in(\mathbb R^d)^T$ is centrally
$(\varepsilon,\delta)$-differentially private (DP) if
\[
\begin{aligned}
&\sup_{\substack{D,D':\, d_H(D,D')=1}}
\sup_{A\subseteq(\mathbb R^d)^T}
\Bigl\{
\mathbb P\{M(D)\in A\}\\
&\qquad\qquad
-e^\varepsilon\mathbb P\{M(D')\in A\}
\Bigr\}
\le \delta.
\end{aligned}
\]
where $d_H$ counts differing records and the second supremum is over
measurable sets $A$. Write $D\sim D'$ when $d_H(D,D')\le1$.
The mechanism accesses the full bank and releases only $M(D)$.
\end{definition}

\begin{definition}[Local and smooth sensitivity \citep{nissim2007smooth}]
\label{def:smooth-sensitivity}
For a query $f:(\mathbb R^d)^n\to\mathbb R^d$, define
\[
    \operatorname{LS}_f(D)=\sup_{D'\sim D}\|f(D)-f(D')\|.
\]
Given a function $L_f$ satisfying $L_f(B)\ge\operatorname{LS}_f(B)$ for every bank $B$, and a parameter $\gamma>0$, define its smooth envelope by
\(
    \mathsf S_\gamma[L_f](D)
    =\sup_{B\in(\mathbb R^d)^n}e^{-\gamma d_H(D,B)}L_f(B),
\)
It satisfies
\[
\begin{aligned}
\operatorname{LS}_f(D) &\le \mathsf S_\gamma[L_f](D),\\
\mathsf S_\gamma[L_f](D) &\le e^\gamma \mathsf S_\gamma[L_f](D'),
\qquad D\sim D'.
\end{aligned}
\]
For $L_f=\operatorname{LS}_f$, this is the $\gamma$-smooth sensitivity.
\end{definition}

\section{Challenges in LSR-DP: Boundary Instability}
\subsection{Retrieval Setup}
For any memory bank
$B=(B_1,\ldots,B_n)\in(\mathbb R^d)^n$, define
\[
I(B,z)=\{i\in[n]:\|B_i-z\|\le r\},
\]
and
\[
\Phi_B(z)=
\begin{cases}
\displaystyle
\frac{1}{|I(B,z)|}\sum_{i\in I(B,z)}B_i,
& I(B,z)\neq\varnothing,\\[2mm]
z,
& I(B,z)=\varnothing.
\end{cases}
\]
For the observed private bank $D=(X_1,\ldots,X_n)$, we write
$I(D,z)$ and $\Phi_D(z)$. The clean trajectory is $z_0=q$ and $z_t=\Phi_D(z_{t-1})$ for
$t=1,\ldots,T$. The update satisfies
\[
    \|\Phi_D(z)-z\|\le r,
    \qquad \|z_t-z_{t-1}\|\le r,
    \qquad \|z_t-q\|\le tr.
\]
Now in the following section, we will present why doing DP for LSR energy will face a fundamentally difficulty challenge compared to LSE energy. 
\subsection{Privacy-Induced Boundary Instability}
\label{subsec:sp-boundary}
The most significant feature about DP in LSR energy is: The LSR update is not merely an average but an average of a subset of the bank fundamentally different than the existing literature studying DP for sample mean or DP for LSE energy. This leads to a central challenge: a small change in the current state can change the active set and produce
an order-$r$ retrieval jump, leading to a substantially large cumulative error after DP. To see this, for the naive recursion
\[
    Y_0^{\mathrm{naive}}=q,
    \qquad
    Y_t^{\mathrm{naive}}=\Phi_D(Y_{t-1}^{\mathrm{naive}})+\xi_t,
\]
the error is
\[
    Y_t^{\mathrm{naive}}-z_t
    =\Phi_D(Y_{t-1}^{\mathrm{naive}})-\Phi_D(z_{t-1})+\xi_t.
\]
This term can remain nonzero for arbitrarily small preceding state error; privacy calibration of $\xi_t$ alone does not control it. The following proposition formally illustrates this phenomenon due to the discontinuity of the update map $\Phi_D$.

\begin{proposition}[Boundary jump]
\label{prop:boundary-non-lipschitz}
Fix $D,z$ and suppose $X_j$ is the unique record satisfying $\|X_j-z\|=r$. Set
\(
z_h=z-hr^{-1}(X_j-z).
\)
For all sufficiently small $h>0$,
\(
I(D,z_h)=I(D,z)\setminus\{j\},
\)
and
\begin{align*}
&\lim_{h\downarrow0}\bigl(\Phi_D(z_h)-\Phi_D(z)\bigr)
\\&=
\begin{cases}
z-X_j, & I(D,z)=\{j\}\\[2mm]
\displaystyle
\frac{1}{|I(D,z)|}
\left(
\frac{\sum_{i\in I(D,z_h)}X_i}{|I(D,z_h)|}-X_j
\right),
& |I(D,z)|\ge2.
\end{cases}
\end{align*}
If this limit is nonzero, then
\[
\frac{\|\Phi_D(z_h)-\Phi_D(z)\|}{\|z_h-z\|}
\longrightarrow\infty,
\]
since $\|z_h-z\|=h\to0$ while the numerator converges to a nonzero constant. Hence $\Phi_D$ is not locally Lipschitz at $z$, and in particular is not uniformly Lipschitz in its state argument.
\end{proposition}
This boundary instability presents a central challenge in LSR-DP and it motivates our proposed method in the following section.

\section{Proposed Method: BRACE}
\label{sec:brace}
BRACE keeps the clean trajectory $(z_t)_{t=0}^T$ internal and releases
only $(Y_t)_{t=1}^T$. Initialize
\(
z_0=Y_0=q.
\)
For $t=1,\ldots,T$, perform the following steps. Let $\gamma_t>0$ and $(\varepsilon_t,\delta_t)$ be public calibration
parameters. For every fixed history $y_{<t}$,
$\nu_t(\cdot;\overline\Delta_t,\gamma_t,\varepsilon_t,\delta_t)$
denotes a centered noise law calibrated so that the conditional
mechanism
\(
M(D) = Q_t(D;y_{<t})+\xi_t
\)
is $(\varepsilon_t,\delta_t)$-DP.

\noindent\textbf{Step 1: Compute retrieval updates.}
Set
\[
\begin{aligned}
I_t &= I(D,z_{t-1}), \qquad J_t = I(D,Y_{t-1}),\\
z_t &= \Phi_D(z_{t-1}), \qquad a_t = \Phi_D(Y_{t-1}).
\end{aligned}
\]
\medskip
\noindent\textbf{Step 2: Compute the correction.}
We choose $R_t\in\mathbb R^d$ by a publicly specified measurable rule applied to $(D,Y_{<t})$ before drawing $\xi_t$. For the selected rule, set
\begin{align}
&R_t\nonumber\\&=
\begin{cases}
\displaystyle
\frac{\sum_{i\in I_t\setminus J_t}(X_i-a_t)
-\sum_{i\in J_t\setminus I_t}(X_i-a_t)}
{|I_t|},
& I_t\neq\varnothing,\\[2mm]
z_{t-1}-a_t,
& I_t=\varnothing.
\end{cases}
\label{eq:active-set-correction}
\end{align}
\medskip
\noindent\textbf{Step 3: Activate the correction.}
Construct $H_t\in\{0,1\}$ as a public measurable function of $(D,Y_{<t})$. With a threshold hyperparameter $\kappa>0$, compute
\begin{equation}
H_t
=
\mathbf 1\{\|R_t\|^2>\kappa T r^2\}.
\label{eq:threshold-selection-rule}
\end{equation}

\medskip
\noindent\textbf{Step 4: Project the corrected query.}
Set
\[
Q_t
=
\Pi_{\overline B_2(Y_{t-1},r)}
\bigl(a_t+H_tR_t\bigr),
\]
where $\Pi$ denotes Euclidean projection onto the closed ball $\overline B_2(Y_{t-1},r)$.

\medskip
\noindent\textbf{Step 5: Compute the sensitivity scale.}
For a bank $B$ and fixed history $y_{<t}$, the local sensitivity of the complete query is
\[
\sup_{B'\sim B}
\left\|
Q_t(B;y_{<t})-Q_t(B';y_{<t})
\right\|.
\]
We construct an upper bound $L_t(B;y_{<t})$ as follows. If $H_s(B';y_{<s})=0$ for all $s\le t$ and all $B'$ with $d_H(B,B')\le1$, set
\[
L_t(B;y_{<t})=\frac{2r}{\max\{1,|I(B,y_{t-1})|\}}.
\]
Otherwise, run \textsc{CorrectedCountSearch}. For each $j\in[n]$ and replacement $x\in\mathbb R^d$, let $B'=B^{j\leftarrow x}$ and recompute Steps~1--4 recursively under $B'$ with $y_{<t}$ fixed. Propagate all clean and released memberships, using
\[
\begin{aligned}
\beta_s^I &= \mathbf 1\{\|x-z_{s-1}(B')\|\le r\},\\
\beta_s^J &= \mathbf 1\{\|x-y_{s-1}\|\le r\}.
\end{aligned}
\]
Enumerate the corresponding membership, activation, and projection branches and retain only feasible branches. Set
\begin{align*}
&L_t(B;y_{<t})
\\&=
\sup_{\substack{
j\in[n],
x\in\mathbb R^d\\
\text{feasible branch}
}}
\left\|
Q_t(B;y_{<t})
-
Q_t(B^{j\leftarrow x};y_{<t})
\right\|.
\end{align*}

Then set $\widehat m_t
=
\frac{2r}{\overline{\Delta}_t(D;y_{<t})}$, where
\[
\begin{aligned}
\overline{\Delta}_t(D;y_{<t})
&=
\max\left\{
\frac{2r}{n},
\mathsf S_{\gamma_t}
\!\left[L_t(\cdot;y_{<t})\right](D)
\right\}.
\end{aligned}
\]

\medskip
\noindent\textbf{Step 6: Add noise, project, and release.}
Draw
\[
\xi_t\sim\nu_t\!\left(\cdot;\overline\Delta_t,\gamma_t,\varepsilon_t,\delta_t\right),
\]
and set
\[
\widetilde Y_t=Q_t+\xi_t,
\qquad
Y_t=\Pi_{\overline B_2(q,(t+1)r)}(\widetilde Y_t).
\]
Release $Y_t$ and use it in the next iteration. Select $\varepsilon>0$ and $0<\delta<1$, and choose the public smoothing parameters according to \eqref{eq:app-gaussian-smoothing}. Set
\[
\xi_t
\sim
\mathcal N\!\left(
0,
\frac{T\overline{\Delta}_t^{\,2}}
{2\rho_{\varepsilon,\delta}}I_d
\right),
\]
where
\(
\rho_{\varepsilon,\delta}
=
\left(
\sqrt{\log(1/\delta)+\varepsilon}
-
\sqrt{\log(1/\delta)}
\right)^2.
\)
The complete algorithm is shown in Algorithm \ref{alg:brace-summary}.

\begin{remark}[Activation, sensitivity, and localization]
\textbf{Activation.}
The threshold $\kappa$ controls how often correction is activated. For the selected rule,
\(
\|Q_t-z_t\|^2
\le
\max\{\kappa Tr^2,\|Y_{t-1}-z_{t-1}\|^2\}.
\)
Larger $\kappa$ gives fewer activations and lower search complexity, but permits a larger residual error.

\textbf{Localization.}
The replacement search may be restricted to
\(
\overline B_2(q,tr)
\cup
\bigcup_{s=1}^t\overline B_2(y_{s-1},r),
\)
together with one outside representative for each replacement index.

\textbf{Sensitivity.}
Since $L_t(B;y_{<t})\le2r$,
\(
\frac{2r}{n}\le\overline{\Delta}_t(D;y_{<t})\le2r,
\qquad
1\le\widehat m_t\le n.
\)
Here $\overline{\Delta}_t$ is a smooth upper bound on the local sensitivity of $Q_t$, and $\widehat m_t=2r/\overline{\Delta}_t$ is an effective sensitivity count rather than an active-set size.
\end{remark}

\begingroup
\begin{algorithm}[t]
\caption{Selected BRACE retrieval}
\label{alg:brace-summary}

\KwIn{Private bank $D$; public $q,r,T,\kappa$; privacy parameters $\varepsilon>0$, $\delta\in(0,1)$; smoothing parameters $(\gamma_t)_{t=1}^T$ satisfying \eqref{eq:app-gaussian-smoothing}.}
\KwOut{Released trajectory $(Y_1,\ldots,Y_T)$.}

$z_0\gets q$, $Y_0\gets q$\;

$\rho_{\varepsilon,\delta}\gets
\left(
\sqrt{\log(1/\delta)+\varepsilon}
-
\sqrt{\log(1/\delta)}
\right)^2$\;

\For{$t=1,\ldots,T$}{

    $I_t\gets I(D,z_{t-1})$,
    $J_t\gets I(D,Y_{t-1})$\;

    $z_t\gets\Phi_D(z_{t-1})$,
    $a_t\gets\Phi_D(Y_{t-1})$\;

    $R_t\gets z_t-a_t$\;

    $H_t\gets
    \mathbf 1\{\|R_t\|^2>\kappa Tr^2\}$\;

    $Q_t\gets
    \Pi_{\overline B_2(Y_{t-1},r)}
    (a_t+H_tR_t)$\;

    Construct $L_t(\cdot;Y_{<t})$ and compute
    $\mathsf S_{\gamma_t}[L_t(\cdot;Y_{<t})](D)$\;

    $\overline\Delta_t\gets
    \max\left\{
    \dfrac{2r}{n},
    \mathsf S_{\gamma_t}[L_t(\cdot;Y_{<t})](D)
    \right\}$\;

    $\widehat m_t\gets\dfrac{2r}{\overline\Delta_t}$\;

    Draw
    $\xi_t\sim
    \mathcal N\!\left(
    0,
    \dfrac{T\overline\Delta_t^{\,2}}
    {2\rho_{\varepsilon,\delta}}I_d
    \right)$\;

    $\widetilde Y_t\gets Q_t+\xi_t$\;
    $Y_t\gets\Pi_{\overline B_2(q,(t+1)r)}(\widetilde Y_t)$\;

    Release $Y_t$\;
}

\Return{$(Y_1,\ldots,Y_T)$}\;
\end{algorithm}
\endgroup

\section{Theoretical Results}
\label{sec:theory}
\subsection{DP Guarantee of BRACE}
For the released trajectory
\(
M(D)=(Y_1,\ldots,Y_T),
\)
define
\(
\delta_{\mathrm{HS}}(\varepsilon)
=
\sup_{D\sim D'}
D_{e^\varepsilon}\!\left(
\mathcal L(M(D))
\,\middle\|\,
\mathcal L(M(D'))
\right).
\)

\begin{theorem}[Central differential privacy of BRACE]
\label{thm:brace-privacy}
For each $t$, conditional on every fixed released history $Y_{<t}=y_{<t}$, suppose the noise law
\(
\xi_t\sim
\nu_t(\cdot;\overline\Delta_t,\gamma_t,\varepsilon_t,\delta_t)
\)
is calibrated so that
\(
Q_t(D;y_{<t})+\xi_t
\)
is $(\varepsilon_t,\delta_t)$-differentially private whenever the sensitivity of $Q_t(\cdot;y_{<t})$ is bounded by $\overline\Delta_t$. If the privacy accountant satisfies
\[
\sum_{t=1}^T\varepsilon_t\le\varepsilon,
\qquad
\sum_{t=1}^T\delta_t\le\delta,
\]
then Algorithm~\ref{alg:brace-summary} is centrally
$(\varepsilon,\delta)$-differentially private. Thus, any conditionally calibrated noise law together with a valid composition accountant yields an $(\varepsilon,\delta)$-DP BRACE trajectory.
\end{theorem}

\subsection{Terminal and Maximum-Path Minimax Risk}
\label{sec:minimax}
Let $\Pall$ denote the class of all probability laws on $\mathbb R^d$.
Throughout this subsection,
$\mathbb E_{D\sim P^n,M}$ denotes expectation jointly over the
i.i.d.\ sample $D\sim P^n$ and the internal randomness of the mechanism
$M$, conditional on the public parameters $(q,r,T)$. For nonnegative $a,b$, write $a\asymp b$ when each is bounded by a finite positive constant times the other, with the dependence of the constants specified below.
For $M(D)=(Y_1,\ldots,Y_T)$, define 
\[
\begin{aligned}
    \Rterm(M)
        &=\sup_{P\in\Pall}\mathbb E_{D\sim P^n,M}\|Y_T-z_T(D)\|^2,\\
    \mathcal R_{\max,T}(M)
        &=\sup_{P\in\Pall}\mathbb E_{D\sim P^n,M}
            \max_{1\le t\le T}\|Y_t-z_t(D)\|^2.
\end{aligned}
\]
Their minimax risks are
\[
\begin{aligned}
    \Rstar(\varepsilon,\delta)
        &=\inf_{M:\,(\varepsilon,\delta)\text{-DP}}\Rterm(M),\\
    \mathcal R_{\max,T}^{\star}(\varepsilon,\delta)
        &=\inf_{M:\,(\varepsilon,\delta)\text{-DP}}
            \mathcal R_{\max,T}(M).
\end{aligned}
\]
The terminal risk measures the final retrieval error, whereas the maximum-path risk controls all released states. The output projection gives a bound uniform in the memory size and dimension.

\begin{theorem}[Minimax-rate optimality]
\label{thm:minimax-optimality}
Let $M$ be a centrally $(\varepsilon,\delta)$-DP BRACE mechanism with the output projection in Step~6, where $\varepsilon\ge0$ and $0\le\delta<1$.

\textup{(i) Fixed horizon.} For fixed $T,\varepsilon,\delta$,
\[
\begin{aligned}
\Rterm(M)
&\asymp\Rstar(\varepsilon,\delta)\asymp r^2,\\
\mathcal R_{\max,T}(M)
&\asymp\mathcal R_{\max,T}^{\star}(\varepsilon,\delta)\asymp r^2.
\end{aligned}
\]
The comparison constants may depend on $(T,\varepsilon,\delta)$ but are independent of $n,d,r$.

\textup{(ii) Growing horizon.} In the region
\begin{equation}
T\ge2,\qquad n\ge2^{T+2}-3,
\label{eq:long-chain-region}
\end{equation}
for fixed $\varepsilon,\delta$,
\[
\begin{aligned}
\Rterm(M)
&\asymp\Rstar(\varepsilon,\delta)\asymp T^2r^2,\\
\mathcal R_{\max,T}(M)
&\asymp\mathcal R_{\max,T}^{\star}(\varepsilon,\delta)\asymp T^2r^2.
\end{aligned}
\]
The comparison constants may depend on $(\varepsilon,\delta)$ but are independent of $n,d,T,r$.
\end{theorem}
\begin{remark} For fixed $T$, increasing $n$ cannot improve the worst-case rate over $\Pall$, and BRACE achieves this optimal $r^2$ rate. The same rate holds for the stronger full-trajectory loss, showing that releasing and controlling the whole path does not worsen the minimax order. When $T$ grows, the matching $T^2r^2$ upper and lower bounds show that BRACE also has the optimal dependence on the retrieval horizon in the long-chain regime.
\end{remark}

\subsection{Trajectory Central Limit Theorems for BRACE}

Throughout this subsection, let $D_n=(X_1,\ldots,X_n)\sim P^n$ for a fixed law $P$, with $q,d,T,r$ fixed as $n\to\infty$. Write $\xi_{n,t}$ for the privacy noise at iteration $t$ for sample size $n$. Convergence is under the joint law of the data and mechanism randomness. For $X\sim P$, define
\[
\bar z_0=q,\qquad
p_t=P(\|X-\bar z_{t-1}\|\le r),
\]
and, whenever $p_t>0$,
\[
\bar z_t=
\frac{\mathbb E[X\mathbf 1\{\|X-\bar z_{t-1}\|\le r\}]}
{p_t}.
\]
For vectors $V_1,\ldots,V_T\in\mathbb R^d$, write
\[
V_{1:T}=(V_1^\top,\ldots,V_T^\top)^\top\in\mathbb R^{dT}.
\]

\begin{assumption}[First-order BRACE conditions]
\label{ass:brace-clt-policy}
There exist random vectors $C_t,\xi_t\in\mathbb R^d$, $t=1,\ldots,T$, such that:

\textup{(i)} Jointly over $t=1,\ldots,T$,
\[
\sqrt n\,(z_{n,t}-a_{n,t}-H_{n,t}R_{n,t})
\Longrightarrow C_t.
\]

\textup{(ii)}
\[
\sqrt n\,\max_{1\le t\le T}
\|Q_{n,t}-(a_{n,t}+H_{n,t}R_{n,t})\|
\xrightarrow{P}0.
\]

\textup{(iii)} Jointly over $t=1,\ldots,T$,
\[
\sqrt n\,\xi_{n,t}\Longrightarrow\xi_t,
\]
where $\mathbb E\xi_t=0$, and the convergence in \textup{(i)} and \textup{(iii)} is joint.
\end{assumption}

Condition~\textup{(i)} controls the residual correction error at the root-$n$ scale, Condition~\textup{(ii)} requires the pre-noise projection to have no first-order effect, and Condition~\textup{(iii)} characterizes the first-order contribution of the privacy noise. These are local asymptotic conditions and do not require global smoothness of the retrieval map.

\begin{theorem}[Trajectory CLT for BRACE]
\label{thm:general-brace-first-order-limit}
Under Assumption~\ref{ass:brace-clt-policy},
\[
\sqrt n\,(Y_{n,1:T}-z_{n,1:T})
\Longrightarrow
-C_{1:T}+\xi_{1:T}.
\]
If $(C_{1:T},\xi_{1:T})$ is jointly Gaussian, then
\[
\begin{aligned}
\sqrt n\,(Y_{n,1:T}-z_{n,1:T})
&\Longrightarrow
\mathcal N_{dT}\!\left(-\mathbb E C_{1:T},\right.\\
&\qquad\left.
\operatorname{Cov}(-C_{1:T}+\xi_{1:T})\right).
\end{aligned}
\]
\end{theorem}

\begin{remark}
\textbf{Cross-time dependence.}
The theorem gives a joint limit for the entire released trajectory rather than only the terminal state. Hence the limiting distribution captures the dependence of retrieval errors across different iterations.

\textbf{Finite-sample and asymptotic guarantees.}
The minimax result provides finite-sample worst-case control of the retrieval error, whereas Theorem~\ref{thm:general-brace-first-order-limit} characterizes its first-order distribution as $n\to\infty$. Together, they describe both the magnitude of the error and its asymptotic uncertainty.

\textbf{Error decomposition.}
The limit
\(
-C_{1:T}+\xi_{1:T}
\)
separates the two first-order sources of error: $C_{1:T}$ represents the residual correction error, while $\xi_{1:T}$ represents the limiting contribution of the privacy noise.
\end{remark}

\begin{theorem}[Effect of the effective count]
\label{thm:effective-count-effect}
For the Gaussian noise in Step~6, let $\varepsilon>0$ and $\delta\in(0,1)$ be fixed as $n\to\infty$. Since
\[
\overline\Delta_{n,t}
=
\frac{2r}{\widehat m_{n,t}},
\]
conditional on $(D_n,Y_{<t})$,
\[
\sqrt n\,\xi_{n,t}
\sim
\mathcal N\!\left(
0,
\frac{2Tr^2}{\rho_{\varepsilon,\delta}}
\frac{n}{\widehat m_{n,t}^{\,2}}I_d
\right).
\]
Consequently, for each fixed $t$,
\[
\sqrt n\,\xi_{n,t}
\begin{cases}
\xrightarrow{P}0,
&
\displaystyle
\frac{\widehat m_{n,t}}{\sqrt n}\xrightarrow{P}\infty,
\\[3mm]
\Longrightarrow
\mathcal N\!\left(
0,
\displaystyle
\frac{2Tr^2\ell_t^2}{\rho_{\varepsilon,\delta}}I_d
\right),
&
\displaystyle
\frac{\sqrt n}{\widehat m_{n,t}}
\xrightarrow{P}\ell_t\in(0,\infty),
\\[5mm]
\text{is not tight},
&
\displaystyle
\frac{\widehat m_{n,t}}{\sqrt n}\xrightarrow{P}0.
\end{cases}
\]
Thus $\widehat m_{n,t}$ of order $\sqrt n$ is the transition regime in which the privacy noise contributes at the root-$n$ scale.

Moreover, if
\[
\frac{\widehat m_{n,t}}{n}
\xrightarrow{P}
c_t>0,
\]
then
\[
\sqrt n\,\xi_{n,t}\xrightarrow{P}0,
\qquad
n\xi_{n,t}
\Longrightarrow
\mathcal N\!\left(
0,
\frac{2Tr^2}
{\rho_{\varepsilon,\delta}c_t^2}I_d
\right).
\]
\end{theorem}

\begin{remark}The effective count $\widehat m_{n,t}=2r/\overline\Delta_{n,t}$ summarizes the sensitivity scale of BRACE. Theorem~\ref{thm:effective-count-effect} determines exactly when privacy noise contributes to the first-order trajectory limit: $\widehat m_{n,t}\asymp\sqrt n$ is the transition between asymptotically negligible, first-order, and non-tight privacy noise at the root-$n$ scale.
\end{remark}

\begin{corollary}[Selected active-set BRACE]
\label{cor:selected-brace-trajectory-limit}
Suppose $p_t>0$ for every $t\le T$ and $P$ has a continuous Lebesgue density in a neighborhood of each sphere
\[
\{x:\|x-\bar z_{t-1}\|=r\}.
\]
With $\kappa,\gamma_1,\ldots,\gamma_T,\varepsilon>0$ and $\delta\in(0,1)$ fixed, Algorithm~\ref{alg:brace-summary} satisfies Assumption~\ref{ass:brace-clt-policy} with $C_t=0$ and $\xi_t=0$ for every $t$. Consequently,
\[
\sqrt n\,(Y_{n,1:T}-z_{n,1:T})
\xrightarrow{P}0.
\]
\end{corollary}

\begin{remark}[Dimension dependence]
For the Gaussian implementation,
\(
\mathbb E[\|\xi_{n,t}\|^2\mid D_n,Y_{<t}]
=
dT\overline\Delta_{n,t}^{\,2}/(2\rho_{\varepsilon,\delta}),
\)
so the unprojected noise energy scales with $d$. However, the final projection gives
\(
\max_{t\le T}\|Y_{n,t}-z_{n,t}\|^2\le(2T+1)^2r^2,
\)
and hence the minimax bounds in Theorem~\ref{thm:minimax-optimality} are dimension independent.
\end{remark}

\section{Experiments}
\label{sec:experiments}

\subsection{Synthetic Data: Boundary-Induced Mode Switching}
\label{subsec:synthetic-boundary}

We first illustrate the boundary instability in Proposition~\ref{prop:boundary-non-lipschitz}. We generate $n=256$ memories in $\mathbb R^2$ from a three-component Gaussian mixture with two dominant modes centered at $(\pm1,0)^\top$ and a broader central component connecting them. We compare noiseless DAM, Naive DAM, and BRACE with $r=0.8$ ($\beta=3.125$), $T=5$, and $(\varepsilon,\delta)=(128,10^{-5})$. The same Gaussian innovations are used for the two noisy methods. This experiment aims to illustrate the boundary instability phenomenon in LSR-DP and demonstrates the necessity for a boundary correction DP method which motivates our BRACE. 

\textbf{Main results.} Figure~\ref{fig:synthetic-boundary} shows that the first perturbation has norm only $0.01443$, but at update~2 it changes 11 gate memberships: six memories leave the clean gate and five enter the released gate, although the gate size changes only from $67$ to $66$. The resulting correction has norm $\|R_2\|=0.12510$ and BRACE activates from update~2 onward. Naive DAM subsequently moves toward the opposite mode: its per-coordinate squared error increases from $0.08445$ at update~3 to $0.42699$ at update~4 and $1.00503$ at update~5, whereas the final BRACE error is $0.01222$. This behavior agrees with Proposition~\ref{prop:boundary-non-lipschitz}: a small perturbation of the state can change the active set and produce a much larger change in the next retrieval update. BRACE corrects this boundary-induced displacement before the next privacy perturbation is added.

\begin{figure}[!t]
    \centering
    \includegraphics[width=\columnwidth]{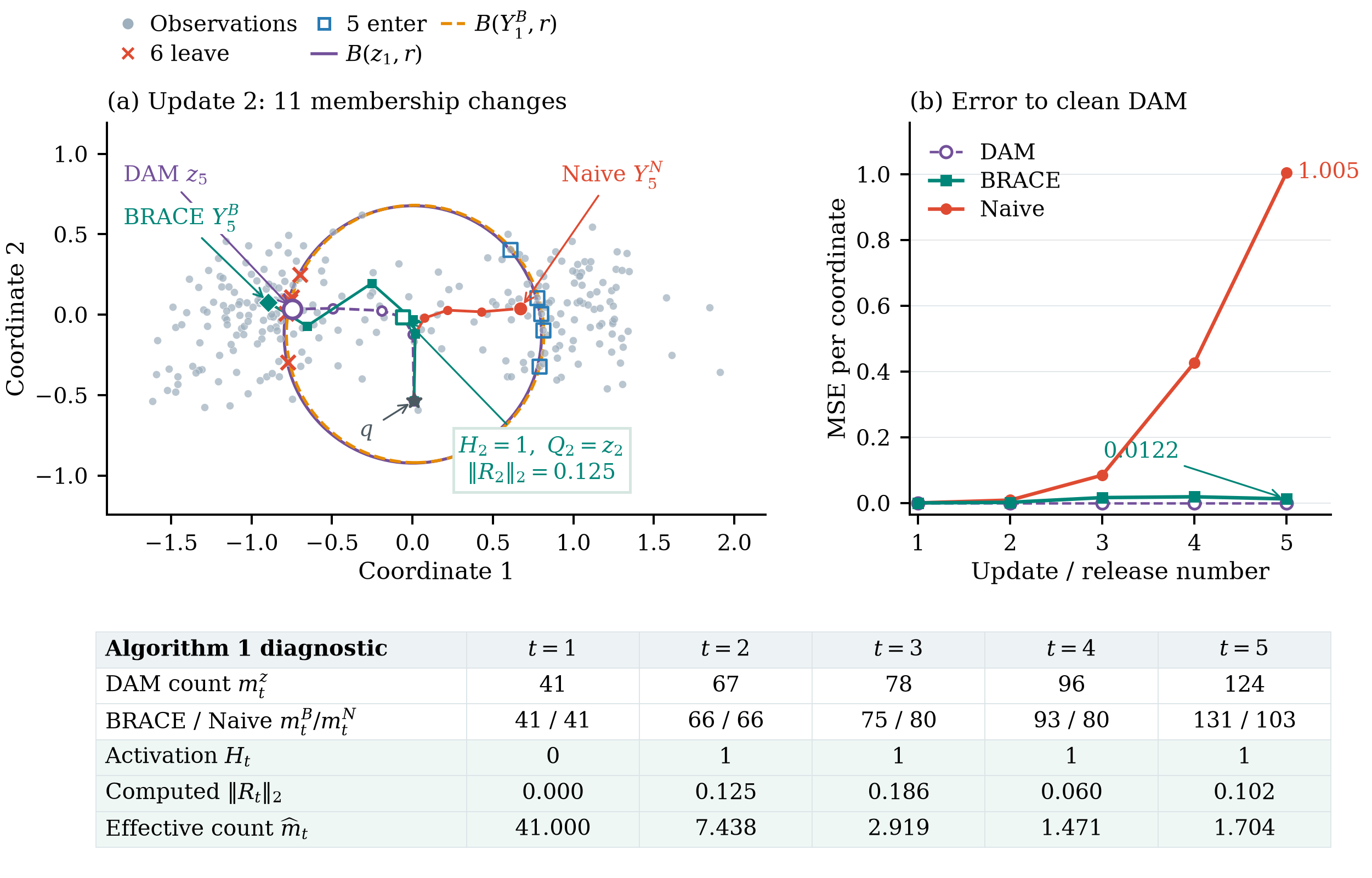}
    \caption{Synthetic boundary instability. Left: retrieval trajectories of noiseless DAM, Naive DAM, and BRACE. Right: per-coordinate squared error relative to noiseless DAM.}
    \label{fig:synthetic-boundary}
\end{figure}

\subsection{MNIST retrieval with strong query corruption}
\label{subsec:mnist-three-plots}
\begin{figure}[!t]
    \centering
    \includegraphics[width=\columnwidth]{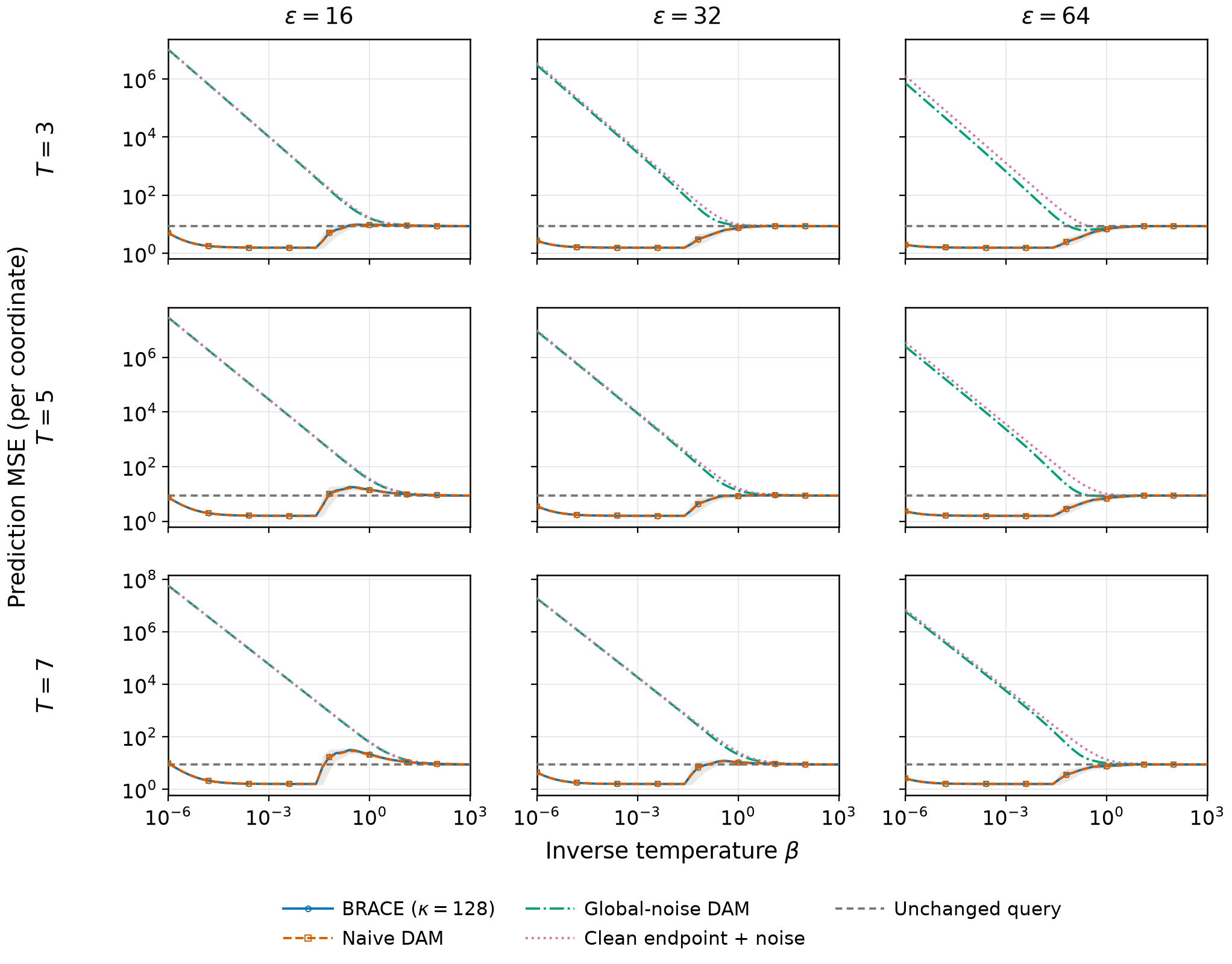}
    \caption{Prediction MSE against uncorrupted query features. Crosses denote unresolved BRACE sensitivity calculations.}
    \label{fig:mnist-prediction-grid}
\end{figure}
\begin{figure}[!t]
    \centering
    \includegraphics[width=\columnwidth]{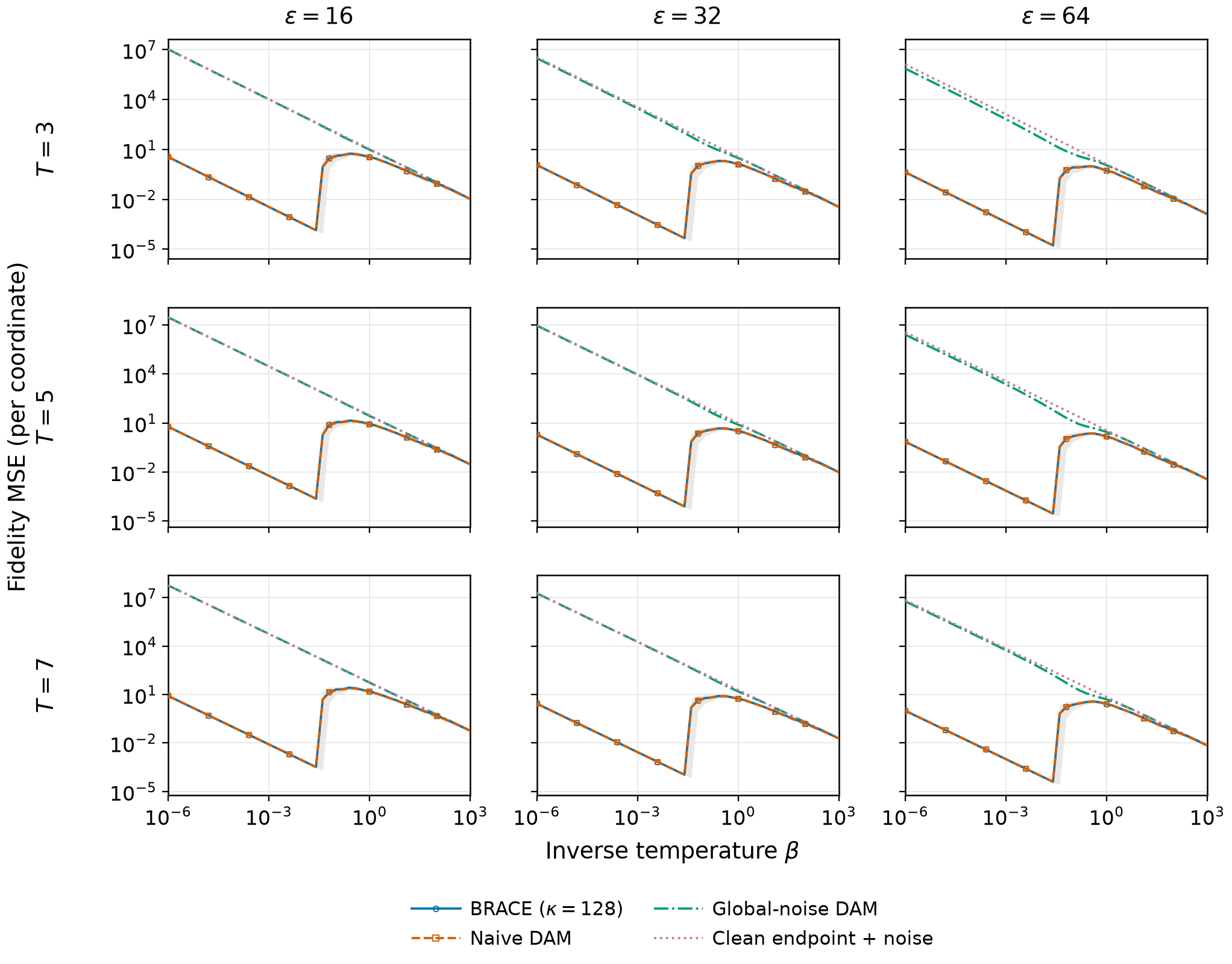}
    \caption{Endpoint fidelity to noiseless DAM. Crosses denote unresolved BRACE sensitivity calculations.}
    \label{fig:mnist-fidelity-grid}
\end{figure}
We use two PCA coordinates of a frozen ten-dimensional VAE trained on 50,000 MNIST images, with a bank of $n=1000$ images and $Q=256$ evaluation queries. Each method starts from the strongly corrupted feature
\[
q_i=x_i+3u_i,\qquad u_i\sim\mathcal N(0,I_2),
\]
with the same corruption across methods. We average 32 paired noise repetitions and vary $\beta\in[10^{-6},1000]$, $T\in\{3,5,7\}$, and $\varepsilon\in\{16,32,64\}$, with $\delta=10^{-5}$ and $r=\sqrt{2/\beta}$. We report prediction MSE relative to the uncorrupted feature $x_i$ and fidelity MSE relative to the noiseless DAM endpoint $z_{i,T}$. Shading gives pointwise 95\% query-bootstrap intervals.

We compare BRACE with $\kappa=128$, Naive DAM, global-noise DAM, and clean endpoint plus Gaussian noise. BRACE uses the smooth sensitivity certificate $\overline\Delta_t$, while Naive DAM uses the direct count scale $2r/\max\{1,m_t\}$ and global-noise DAM uses $2r$. The iterative methods add Gaussian noise with standard deviation $\Delta_t\sqrt{T/(2\rho)}$, where
\(
\rho=
\left(
\sqrt{\log(1/\delta)+\varepsilon}
-
\sqrt{\log(1/\delta)}
\right)^2.
\)
The clean-endpoint baseline releases only $z_T$ with one Gaussian perturbation calibrated to the bound $2Tr$.

\paragraph{Main result.}
Here the best setting refers to the evaluated $\beta$ giving the smallest prediction MSE. BRACE achieves its minimum prediction MSE $1.582963$ at $\beta=0.025119$, $T=3$, and $\varepsilon=64$, compared with $8.759678$ for the corrupted input and $1.582947$ for noiseless DAM. Thus BRACE reduces the initial prediction error by $81.9\%$ and nearly matches noiseless retrieval. Naive DAM attains the same measured minimum, whereas the minima of global-noise DAM and clean-endpoint noise are $6.327462$ and $7.870764$, respectively.
Since $r=\sqrt{2/\beta}$, increasing $\beta$ reduces the number of memories in each gate. 

Interestingly, we also observed that for small $\beta$, the gate contains most or all memories, so retrieval approaches the bank mean; decreasing $\beta$ further mainly increases privacy noise. At intermediate $\beta$, smaller active counts increase the sensitivity scale and can produce the nonmonotone fidelity behavior. For large $\beta$, gates are often empty, so both retrieval movement and privacy noise vanish and the output remains close to the corrupted query.

\paragraph{Handwritten reconstructions.}
Figure~\ref{fig:mnist-best-brace-digits} provides a visual illustration at the best observed BRACE setting. BRACE closely follows noiseless DAM, while the reconstructions also show that low feature-space MSE does not necessarily imply recovery of the original individual digit.

\begin{figure}[!t]
    \centering
    \includegraphics[width=\columnwidth]{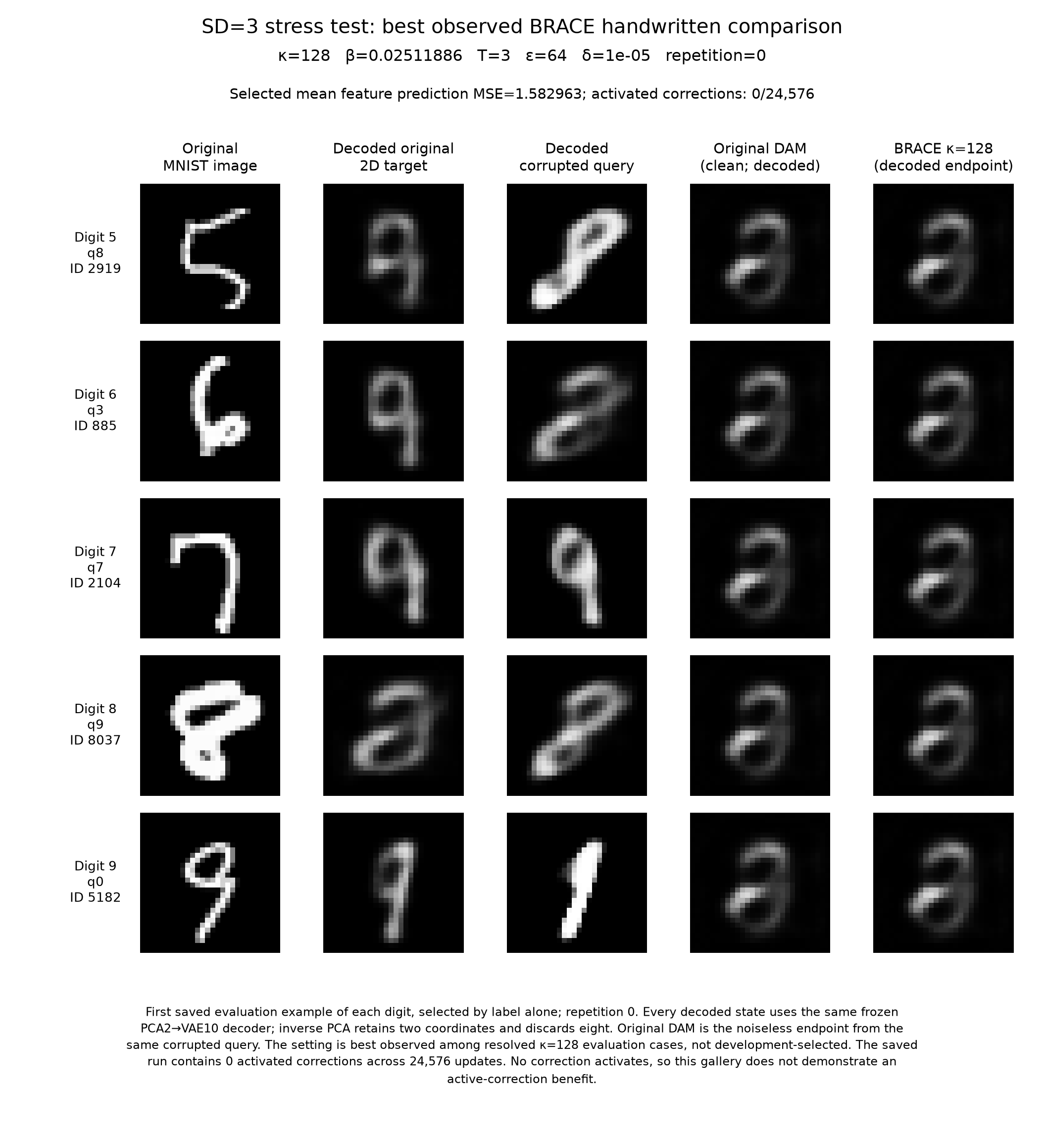}
    \caption{Handwritten examples at the best observed BRACE setting
    ($\kappa=128$, $\beta=0.025119$, $T=3$, $\varepsilon=64$).}
    \label{fig:mnist-best-brace-digits}
\end{figure}

\section{Conclusion and Limitation}
We developed a comprehensive framework for privacy-preserving retrieval of LSR-DAM, establishing formal privacy guarantees, utility bounds, and minimax optimality results. By tackling the boundary instability, one fundamental challenge in this problem, we propose the BRACE method. Our theory characterizes the effects of local retrieval, sensitivity smoothing, and privacy noise, and provides a principled understanding of the resulting privacy-utility trade-off. Experiments demonstrate that the proposed method achieves consistently better prediction performance than competing approaches, while supporting the theoretical findings and showing the practical benefits of local retrieval. Moreover, as with any retrieval-based private mechanism, the practical performance depends on the underlying data geometry and the choice of retrieval scale. In regimes where local neighborhoods are intrinsically sparse, the benefits of retrieval may naturally diminish. The sensitivity estimation also introduces additional computation, reflecting the cost of obtaining tighter data-dependent privacy guarantees. These considerations primarily concern the operating regime and implementation efficiency rather than the validity of the theoretical guarantees. Future work may further improve computational efficiency and develop more adaptive choices of the retrieval scale.

\subsection*{AI use statement}
In this work, we used generative AI tools to assist with proof writing, code debugging, literature search and language editing. We have not used generative AI tools for brainstorming research ideas, problem formulation, algorithm design, finding theoretical insights, and experimental design. The rest
of the required disclosure tasks are not applicable to
this work. We have reviewed all AI-assisted work. For example, LLM-generated code was verified and tested for correctness by two authors; all LLM-assisted proofs were checked line by line by the authors. We take responsibility for the final content of this work, including text, claims, or artifacts produced with the aid of generative AI.

\def\BRACESkipAppendixReferences{1}

\clearpage
\appendix

This supplementary contains proofs of all theoretical results, experiment implement details and computational complexity.

\section{Proof of Proposition~\ref{prop:boundary-non-lipschitz}}
\label{app:geometry}

\begin{proof}
By definition,
\[
\|z_h-z\|=h,
\qquad
\|X_j-z_h\|=r+h.
\]
For $i\ne j$, $\|X_i-z\|\ne r$. By continuity, for all small $h>0$,
\[
I(D,z_h)=I(D,z)\setminus\{j\}.
\]
If $I(D,z)=\{j\}$, then
\[
\Phi_D(z_h)-\Phi_D(z)=z_h-X_j\longrightarrow z-X_j.
\]
For $m=|I(D,z)|\ge2$,
\[
\begin{aligned}
\Phi_D(z_h)-\Phi_D(z)
&=\frac{\sum_{i\in I(D,z)\setminus\{j\}}X_i}{m-1}
  -\frac{X_j+\sum_{i\in I(D,z)\setminus\{j\}}X_i}{m}\\
&=\frac1m\left(
\frac{\sum_{i\in I(D,z)\setminus\{j\}}X_i}{m-1}-X_j
\right).
\end{aligned}
\]
For a nonzero limit,
\[
\frac{\|\Phi_D(z_h)-\Phi_D(z)\|}{\|z_h-z\|}
=\frac{\|\Phi_D(z_h)-\Phi_D(z)\|}{h}\longrightarrow\infty.
\]
\end{proof}

\section{Proofs for the BRACE construction}
\label{app:post-correction-sensitivity}
\label{app:correction-proofs}

\subsection{Movement and fixed-center sensitivity}
\label{app:direct-sensitivity}

\begin{proof}
By convexity, $\Phi_D(y)\in\overline B_2(y,r)$, also for an empty gate. The triangle inequality gives
\[
\|\Phi_D(y)-y\|\le r,
\qquad
\|z_t-z_{t-1}\|\le r,
\qquad
\|z_t-q\|\le tr.
\]

Fix $D'=D^{j\leftarrow x}$ and $m=|I(D,y)|$. For $m=0$, $\|\Phi_{D'}(y)-\Phi_D(y)\|\le r$. For $m\ge1$, an inactive-to-inactive replacement leaves $\Phi_D(y)$ unchanged. For an active-to-active replacement,
\[
\Phi_{D'}(y)-\Phi_D(y)=\frac{x-X_j}{m},
\qquad
\|x-X_j\|\le2r.
\]
For removal with $m\ge2$,
\[
\Phi_D(y)-\Phi_{D'}(y)
=\frac1m\left(
X_j-\frac{\sum_{i\in I(D,y)\setminus\{j\}}X_i}{m-1}
\right),
\]
with norm at most $2r/m$. For $m=1$, removal changes the output by at most $r$. For entry,
\[
\Phi_{D'}(y)-\Phi_D(y)
=\frac{x-\Phi_D(y)}{m+1},
\]
with norm at most $2r/(m+1)\le2r/m$. Thus
\begin{equation}
\operatorname{LS}_{\Phi_{\cdot}(y)}(D)
\le\frac{2r}{\max\{1,|I(D,y)|\}}\le2r.
\label{eq:app-fixed-center-bound}
\end{equation}
\end{proof}

\subsection{Selected correction and activation bound}
\label{app:correction-residual}

\begin{proof}
For $I_t\ne\varnothing$,
\[
\begin{aligned}
&\sum_{i\in I_t\setminus J_t}(X_i-a_t)
 -\sum_{i\in J_t\setminus I_t}(X_i-a_t)\\
&\qquad=\sum_{i\in I_t}(X_i-a_t)
 -\sum_{i\in J_t}(X_i-a_t)
 =|I_t|(z_t-a_t).
\end{aligned}
\]
Here $\sum_{i\in J_t}(X_i-a_t)=0$, also for $J_t=\varnothing$. For $I_t=\varnothing$, $z_t=z_{t-1}$. Thus $R_t=z_t-a_t$ in both cases.

For $H_t=0$, $a_t\in\overline B_2(Y_{t-1},r)$ and $Q_t=a_t$, so
\[
\|Q_t-z_t\|^2=\|R_t\|^2\le\kappa Tr^2.
\]
For $H_t=1$, $a_t+H_tR_t=z_t$, so
\[
\begin{aligned}
\|Q_t-z_t\|
&=(\|z_t-Y_{t-1}\|-r)_+\\
&\le\bigl(\|z_t-z_{t-1}\|
       +\|z_{t-1}-Y_{t-1}\|-r\bigr)_+\\
&\le\|z_{t-1}-Y_{t-1}\|.
\end{aligned}
\]
Thus
\[
\|Q_t-z_t\|^2
\le\max\{\kappa Tr^2,\|Y_{t-1}-z_{t-1}\|^2\}.
\]
\end{proof}

\subsection{Replacement branches, localization, and raw sensitivity}
\label{app:branch-search}

\begin{proof}
Fix $B$, $j\in[n]$, and $y_{<t}$. Set $B'_j=x$ and $B'_i=B_i$ for $i\ne j$. Primes denote BRACE on $B'=B^{j\leftarrow x}$ at the fixed history.

For $s\le t$, enumerate $I'_s,J'_s\subseteq[n]$, $H'_s\in\{0,1\}$, and both projection cases. Set $z'_0=q$ and require
\[
\begin{aligned}
i\in I'_s&\ \Longleftrightarrow\ \|B'_i-z'_{s-1}\|^2\le r^2,\\
i\in J'_s&\ \Longleftrightarrow\ \|B'_i-y_{s-1}\|^2\le r^2,\\
z'_s&=\begin{cases}
|I'_s|^{-1}\sum_{i\in I'_s}B'_i,&I'_s\ne\varnothing,\\
z'_{s-1},&I'_s=\varnothing,
\end{cases}\\
a'_s&=\begin{cases}
|J'_s|^{-1}\sum_{i\in J'_s}B'_i,&J'_s\ne\varnothing,\\
y_{s-1},&J'_s=\varnothing,
\end{cases}\\
R'_s&=z'_s-a'_s,\qquad
H'_s=\mathbf1\{\|R'_s\|^2>\kappa Tr^2\}.
\end{aligned}
\]
Inside the projection ball, require
\[
\|a'_s+H'_sR'_s-y_{s-1}\|^2\le r^2,
\qquad
Q'_s=a'_s+H'_sR'_s.
\]
Outside, use $\eta_s>r$ and require
\[
\begin{aligned}
\eta_s^2&=\|a'_s+H'_sR'_s-y_{s-1}\|^2,\\
\eta_s(Q'_s-y_{s-1})
&=r(a'_s+H'_sR'_s-y_{s-1}).
\end{aligned}
\]
These polynomial constraints include all records at every step.

Every replacement satisfies a branch. Conversely, the constraints and $z'_0=q$ determine $(z'_s,a'_s,R'_s,H'_s,Q'_s)$ by induction. Hence, whenever \textsc{CorrectedCountSearch} is used, the branch supremum equals the complete-query local sensitivity.

If $x$ is in a clean gate at $s\le t$,
\[
\|x-q\|
\le\|x-z'_{s-1}\|+\|z'_{s-1}-q\|
\le sr\le tr.
\]
A released gate requires $x\in\overline B_2(y_{s-1},r)$. Outside
\[
\overline B_2(q,tr)
\ \cup\ \bigcup_{s=1}^t\overline B_2(y_{s-1},r),
\]
all gates exclude $x$. For fixed $j$, these replacements give the same averages and recursion. One outside representative suffices.

Both queries lie in $\overline B_2(y_{t-1},r)$, so
\[
\|Q_t(B;y_{<t})-Q_t(B';y_{<t})\|\le2r.
\]
Under the nonactivation condition in Step~5, use \eqref{eq:app-fixed-center-bound}; otherwise, use the branch supremum. Thus
\begin{equation}
\operatorname{LS}_{Q_t(\cdot;y_{<t})}(B)
\le L_t(B;y_{<t})\le2r.
\label{eq:app-raw-certificate-bound}
\end{equation}
\end{proof}

\paragraph{Computational complexity.}
Let
\[
s_1<\cdots<s_{K_t}\le t,
\qquad
s_{K_t+1}=t+1,
\]
be the activated correction times. Before $s_1$, no correction branch search is needed. At the $k$th activation, the cumulative neighboring-query discrepancy is restricted to
\[
\overline B_2(0,2rk).
\]
Let $G_k$ be the number of feasible gate cells retained at activation $k$. The number of branch-propagation operations is bounded by
\[
\sum_{k=1}^{K_t}
(s_{k+1}-s_k)
\prod_{\ell=1}^kG_\ell
\le
(t-s_1+1)\prod_{k=1}^{K_t}G_k.
\]
If $K_t=0$, no branch search is required. Thus the branch complexity depends on the number of activated corrections, rather than branching over the full bank at every iteration.

\subsection{Smooth domination and measurability}
\label{app:smooth-certificate}

\begin{proof}
By \eqref{eq:app-raw-certificate-bound} and the envelope term $B=D$,
\[
\operatorname{LS}_{Q_t(\cdot;y_{<t})}(D)
\le L_t(D;y_{<t})
\le\mathsf S_{\gamma_t}[L_t](D).
\]
Every envelope term is at most $2r$. For $D\sim D'$, the Hamming triangle inequality gives
\[
\begin{aligned}
e^{-\gamma_t d_H(D,B)}L_t(B;y_{<t})
&\le e^{\gamma_t}
 e^{-\gamma_t d_H(D',B)}L_t(B;y_{<t}).
\end{aligned}
\]
Take suprema and include the floor $2r/n$:
\begin{equation}
\begin{aligned}
\frac{2r}{n}&\le\overline\Delta_t(D;y_{<t})\le2r,\\
\overline\Delta_t(D;y_{<t})
&\le e^{\gamma_t}\overline\Delta_t(D';y_{<t}),\\
1&\le\widehat m_t\le n.
\end{aligned}
\label{eq:app-smooth-certificate-bound}
\end{equation}
For $c\ge0$,
\[
a\le e^{\gamma_t}b
\quad\Longrightarrow\quad
\max\{c,a\}\le e^{\gamma_t}\max\{c,b\}.
\]

The branch graphs form a finite semialgebraic union. By quantifier elimination \citep{basu2006algorithms}, $Q_t$ and the all-neighbor nonactivation condition are semialgebraic.

For $u\ge0$, local sensitivity exceeds $u$ exactly when a replacement satisfies
\[
\|Q_t(B;y_{<t})-Q_t(B';y_{<t})\|^2>u^2.
\]
Quantifier elimination makes this relation semialgebraic. The count branch is also semialgebraic, so $L_t$ is Borel measurable.

Also,
\begin{equation}
\mathsf S_{\gamma_t}[L_t](D)
=\max_{0\le k\le n}
 e^{-\gamma_t k}\sup_{d_H(D,B)\le k}L_t(B;y_{<t}).
\label{eq:app-finite-radius-envelope}
\end{equation}
For $d_H(D,B)=j\le k$, $e^{-\gamma_t k}L_t(B)\le e^{-\gamma_t j}L_t(B)$; take $k=j$ for the reverse bound. Each supremum has semialgebraic strict superlevel sets by quantifier elimination. Finite maxima and fixed weights preserve Borel measurability.
For real-algebraic inputs, quantifier elimination and bisection on $[0,4r^2]$ give certified bounds for the required suprema.
\end{proof}

\section{Proof of Theorem~\ref{thm:brace-privacy}}
\label{app:privacy-proof}

\begin{proof}
For fixed $y_{<t}$ and Borel $A$, define the release kernel
\[
\begin{aligned}
K_t^D(A\mid y_{<t})
=\mathbb P\bigl\{
\Pi_{\overline B_2(q,(t+1)r)}
\bigl(Q_t(D;y_{<t})+\xi_t\bigr)\in A
\bigr\}.
\end{aligned}
\]
Apply conditional DP to the preimage of $A$ under the public projection.

Fix ordered neighbors $D,D'$. Set
\[
\mu_t(\cdot\mid y_{<t})
=\tfrac12\bigl(K_t^D(\cdot\mid y_{<t})
               +K_t^{D'}(\cdot\mid y_{<t})\bigr),
\]
The standard Borel spaces admit jointly measurable densities $k_t^D,k_t^{D'}$ relative to $\mu_t$. Conditional DP gives
\[
\int(k_t^D-e^{\varepsilon_t}k_t^{D'})_+\,d\mu_t\le\delta_t.
\]
Set $\widetilde k_t=\min\{k_t^D,e^{\varepsilon_t}k_t^{D'}\}$. Then $\int\widetilde k_t\,d\mu_t\ge1-\delta_t$. With
\[
\mu(dy_{1:T})=\prod_{t=1}^T\mu_t(dy_t\mid y_{<t})
\]
defined by iterated integration and $y_0=q$, backward integration gives
\[
\int\prod_{t=1}^T\widetilde k_t\,d\mu
\ge\prod_{t=1}^T(1-\delta_t)
\ge1-\sum_{t=1}^T\delta_t.
\]

\[
\prod_{t=1}^T\widetilde k_t
\le\prod_{t=1}^Tk_t^D,
\qquad
\prod_{t=1}^T\widetilde k_t
\le e^{\sum_t\varepsilon_t}\prod_{t=1}^Tk_t^{D'}.
\]
For every measurable $A$,
\[
\begin{aligned}
\mathbb P\{M(D)\in A\}
&\le e^{\sum_t\varepsilon_t}\mathbb P\{M(D')\in A\}
 +\sum_t\delta_t\\
&\le e^\varepsilon\mathbb P\{M(D')\in A\}+\delta.
\end{aligned}
\]
This holds for all ordered neighbors \citep{dwork2014algorithmic}.
\end{proof}

\subsection{Exact deficit for the released trajectory}
\label{app:exact-accountant}

\begin{proof}
For probability measures $P,Q$, define
\[
D_{e^\varepsilon}(P\|Q)
=\sup_{A\text{ measurable}}\{P(A)-e^\varepsilon Q(A)\}.
\]
The transcript densities relative to $\mu$ are $\prod_t k_t^D$ and $\prod_t k_t^{D'}$. Taking the positive set gives
\[
\begin{aligned}
\delta_{\mathrm{HS}}(\varepsilon)
&=\sup_{D\sim D'}\int
\left(\prod_{t=1}^T k_t^D
-e^\varepsilon\prod_{t=1}^T k_t^{D'}\right)_+\,d\mu.
\end{aligned}
\]
Thus $M$ is $(\varepsilon,\delta)$-DP exactly when $\delta_{\mathrm{HS}}(\varepsilon)\le\delta$. The measure $\mu$ includes projection boundary mass.
\end{proof}

\subsection{Gaussian calibration in Step~6}
\label{app:smooth-noise-calibration}

\begin{proof}
Let $\varepsilon>0$, $0<\delta<1$, and set
\[
\alpha=1+\sqrt{\frac{\log(1/\delta)}{\rho_{\varepsilon,\delta}}}>1.
\]
Choose
\begin{equation}
\begin{aligned}
0<\gamma_t
&\le\min\left\{\frac14,\frac1{8(\alpha-1)}\right\},\\
\gamma_t
&\le\frac{\alpha\log\alpha-(\alpha-1)\log(\alpha-1)}
{(\alpha-1)\{8\alpha(\alpha-1)\rho_{\varepsilon,\delta}+5Td\}}.
\end{aligned}
\label{eq:app-gaussian-smoothing}
\end{equation}

For densities $p,p'$ under a common measure, define
\[
D_\alpha(p\|p')
=\frac1{\alpha-1}\log\int p^\alpha(p')^{1-\alpha}.
\]
Take $D_\alpha=\infty$ if absolute continuity or integrability fails. A uniform conditional bound $\eta_t$ gives
\[
\mathbb E_{D'}\!\left[
\left(\frac{k_t^D(Y_t\mid Y_{<t})}
{k_t^{D'}(Y_t\mid Y_{<t})}\right)^\alpha
\Bigm|Y_{<t}\right]
\le e^{(\alpha-1)\eta_t}.
\]
Backward integration bounds the transcript divergence by $\sum_t\eta_t$ \citep{mironov2017renyi}.

For $a>0$,
\[
\sup_{u>a}\frac{u-a}{u^\alpha}
=\frac{(\alpha-1)^{\alpha-1}}{\alpha^\alpha a^{\alpha-1}},
\]
The derivative is $u^{-\alpha-1}[\alpha a-(\alpha-1)u]$. Set $a=e^\varepsilon$ and integrate:
\begin{equation}
D_{e^\varepsilon}(p\|p')
\le\frac{(\alpha-1)^{\alpha-1}}{\alpha^\alpha}
 e^{(\alpha-1)(D_\alpha(p\|p')-\varepsilon)}.
\label{eq:app-renyi-conversion}
\end{equation}
See \citet{balle2020hypothesis}.

For Gaussian densities with means $u,u'$ and variances $\sigma^2I_d,\sigma'^2I_d$, set $v=\sigma^2/\sigma'^2$. Completing the square gives, for $\alpha-(\alpha-1)v>0$,
\[
\begin{aligned}
\int p^\alpha(p')^{1-\alpha}
&=v^{-d(\alpha-1)/2}
 [\alpha-(\alpha-1)v]^{-d/2}\\
&\quad\times\exp\left\{
\frac{\alpha(\alpha-1)\|u-u'\|^2}
{2\sigma'^2[\alpha-(\alpha-1)v]}
\right\}.
\end{aligned}
\]
Under $y=u+\sigma x$, the quadratic coefficient is $-[\alpha-(\alpha-1)v]/2$; the integral is infinite otherwise. Taking logarithms gives
\begin{equation}
\begin{aligned}
&D_\alpha\!\left(
\mathcal N(u,\sigma^2I_d)\middle\|\mathcal N(u',\sigma'^2I_d)
\right)\\
&\quad=\frac{\alpha\|u-u'\|^2}
{2\sigma'^2[\alpha-(\alpha-1)v]}
-\frac d2\log v\\
&\qquad\quad-\frac{d}{2(\alpha-1)}
\log[\alpha-(\alpha-1)v].
\end{aligned}
\label{eq:app-unequal-gaussian-divergence}
\end{equation}

Fix neighbors and a history. Let $s,s'$ be their certificates and $u,u'$ their queries. By \eqref{eq:app-smooth-certificate-bound},
\[
\|u-u'\|\le\min\{s,s'\},
\qquad
e^{-\gamma_t}\le s/s'\le e^{\gamma_t}.
\]
Using $\sigma^2=Ts^2/(2\rho_{\varepsilon,\delta})$ and $\sigma'^2=Ts'^2/(2\rho_{\varepsilon,\delta})$, \eqref{eq:app-unequal-gaussian-divergence} is at most
\[
\begin{aligned}
&\frac{\alpha\rho_{\varepsilon,\delta}/T}
{1-(\alpha-1)(e^{2\gamma_t}-1)}+d\gamma_t\\
&\qquad-\frac{d}{2(\alpha-1)}
\log\{1-(\alpha-1)(e^{2\gamma_t}-1)\}.
\end{aligned}
\]
Here $\|u-u'\|\le s'$, $v\le e^{2\gamma_t}$, and $|\log v|\le2\gamma_t$.

For $0\le\gamma\le1/4$,
\[
e^{2\gamma}-1=\int_0^\gamma2e^{2x}\,dx\le4\gamma.
\]
By \eqref{eq:app-gaussian-smoothing},
\[
0\le(\alpha-1)(e^{2\gamma_t}-1)\le\frac12.
\]
Using $(1-x)^{-1}\le1+2x$ and $-\log(1-x)\le2x$ for $0\le x\le1/2$, the divergence is at most
\[
\frac{\alpha\rho_{\varepsilon,\delta}}{T}
+\left\{
\frac{8\alpha(\alpha-1)\rho_{\varepsilon,\delta}}{T}+5d
\right\}\gamma_t.
\]
The projected likelihood ratio is the conditional mean of the input ratio. Jensen's inequality for $x^\alpha$ gives the same divergence bound.

Composition and \eqref{eq:app-gaussian-smoothing} bound the transcript divergence by
\[
\alpha\rho_{\varepsilon,\delta}
+\frac{\alpha\log\alpha-(\alpha-1)\log(\alpha-1)}{\alpha-1}.
\]
By definition,
\[
\begin{aligned}
\varepsilon
&=\rho_{\varepsilon,\delta}
 +2\sqrt{\rho_{\varepsilon,\delta}\log(1/\delta)},\\
(\alpha-1)(\alpha\rho_{\varepsilon,\delta}-\varepsilon)
&=-\log(1/\delta).
\end{aligned}
\]
Equation~\eqref{eq:app-renyi-conversion} gives
\[
\delta_{\mathrm{HS}}(\varepsilon)
\le e^{-\log(1/\delta)}=\delta.
\]

For per-step budgets, use $T=1$ and $(\varepsilon_t,\delta_t)$ above. Variance $\overline\Delta_t^2/(2\rho_{\varepsilon_t,\delta_t})$ and the corresponding smoothing bound give conditional DP; apply Theorem~\ref{thm:brace-privacy}.
\end{proof}

\section{Proof of Theorem~\ref{thm:minimax-optimality}}
\label{app:minimax-proofs}

\begin{proof}
\textbf{Output-projection bound.}
By Appendix~\ref{app:direct-sensitivity}, $z_t\in\overline B_2(q,tr)$. Projection is nonexpansive, so
\[
\|Y_t-z_t\|
\le\|\widetilde Y_t-z_t\|.
\]
The triangle inequality gives
\[
\|Y_t-z_t\|
\le\|Y_t-q\|+\|z_t-q\|
\le(2t+1)r.
\]
Thus, for every bank and noise draw,
\begin{equation}
\begin{aligned}
\|Y_t-z_t\|
&\le\min\{\|\widetilde Y_t-z_t\|,(2t+1)r\},\\
\Rterm(M)
&\le\mathcal R_{\max,T}(M)\le(2T+1)^2r^2.
\end{aligned}
\label{eq:app-output-projection-bound}
\end{equation}

\par\medskip\noindent\textbf{Complete-neighbor lower bound.}
Let $K_x$, $x\in\overline B_2(q,r)$, be measurable kernels with
\[
K_x(A)\le e^\varepsilon K_{x'}(A)+\delta
\qquad\text{for every }x,x',A.
\]
Write
\[
\mathcal R=\sup_{x\in\overline B_2(q,r)}
\int\|u-x\|^2K_x(du).
\]
We may assume $\mathcal R<\infty$. If $\mathcal R=0$, all outputs equal $x$, contradicting $\delta<1$. Set
\[
h=\sqrt{\frac{2\mathcal R}{1-\delta}}.
\]
If $h\ge r/4$, then $\mathcal R\ge(1-\delta)r^2/32$.

For $h<r/4$, translate $q$ to zero. Fix a unit vector $e_1$ and $x_0=-re_1/2$. For $x\in\overline B_2(re_1/2,r/4)$, Markov and DP give
\[
\begin{aligned}
K_x(\overline B_2(x,h))
&\ge1-\frac{\mathcal R}{h^2}
 =1-\frac{1-\delta}{2},\\
K_{x_0}(\overline B_2(x,h))
&\ge e^{-\varepsilon}
 \bigl(K_x(\overline B_2(x,h))-\delta\bigr)\\
&\ge\frac{1-\delta}{2}e^{-\varepsilon}.
\end{aligned}
\]
Here $\|u-x\|\le h$ implies $\|u-x_0\|>r/2$. Let $V_d$ be the unit-ball volume. By Tonelli and Markov,
\[
\begin{aligned}
V_d(r/4)^d\frac{1-\delta}{2}e^{-\varepsilon}
&\le\int_{\overline B_2(re_1/2,r/4)}
 K_{x_0}(\overline B_2(x,h))\,dx\\
&\le V_dh^d K_{x_0}\{u:\|u-x_0\|\ge r/2\}\\
&\le V_dh^d\frac{4\mathcal R}{r^2}.
\end{aligned}
\]
Cancel $V_d$ and substitute $h$:
\[
\mathcal R^{(d+2)/2}
\ge2^{-3-5d/2}r^{d+2}
 (1-\delta)^{(d+2)/2}e^{-\varepsilon}.
\]
Using $(5d+6)/(d+2)\le5$, both cases give
\begin{equation}
\sup_x\int\|u-x\|^2K_x(du)
\ge\frac{1-\delta}{32}r^2e^{-2\varepsilon/(d+2)}.
\label{eq:app-location-lower}
\end{equation}

\par\medskip\noindent\textbf{IID rare-record lower bound.}
Fix a unit vector $e_1$ and $b=q+4re_1$. Let $D_{i,x}$ contain $x\in\overline B_2(q,r)$ at position $i$ and $b$ elsewhere. Since $\|b-x\|\ge3r$,
\[
z_t(D_{i,x})=x\qquad(t\ge1).
\]
Let $K_{i,x}$ be the terminal-output law of a private mechanism on $D_{i,x}$, and $K_x=n^{-1}\sum_iK_{i,x}$. Since $D_{i,x}\sim D_{i,x'}$, terminal post-processing and averaging give the DP condition in \eqref{eq:app-location-lower}.

Write $\delta_x$ for the point mass at $x$. For $n\ge2$, take
\[
P_x=(1-1/n)\delta_b+(1/n)\delta_x.
\]
Exactly one $x$ has probability $(1-1/n)^{n-1}$ and uniform position. Thus
\[
\mathbb E_{D\sim P_x^n,M}\|Y_T-z_T(D)\|^2
\ge(1-1/n)^{n-1}\int\|u-x\|^2K_x(du).
\]
For $0<u<1$,
\[
\log(1-u)=-\int_0^u\frac{ds}{1-s}\ge-\frac{u}{1-u},
\]
Thus $(1-1/n)^{n-1}\ge e^{-1}$. Apply \eqref{eq:app-location-lower}, take $\sup_x$, then $\inf_M$:
\begin{equation}
\Rstar(\varepsilon,\delta)
\ge\frac{1-\delta}{32e}r^2e^{-2\varepsilon/(d+2)}
\ge\frac{1-\delta}{32e}r^2e^{-2\varepsilon/3}.
\label{eq:app-rare-lower}
\end{equation}
Here $d\ge1$ and $\varepsilon\ge0$. For $n=1$, use $P_x=\delta_x$. Since terminal loss is at most maximum-path loss, the latter has the same lower bound.

\par\medskip\noindent\textbf{Long-chain hard family.}
Suppose $T\ge2$ and $n\ge2^{T+2}-3$. Set
\[
\begin{aligned}
x_j&=q+\frac{(j+1)r}{2}e_1,
 &&1\le j\le T+1,\\
x_*&=x_1=q+re_1,\\
P_0\{x_j\}&=\frac{2^j}{n-1},
 &&2\le j\le T+1.
\end{aligned}
\]
Put the remaining mass at $q+(T+3)re_1$; this is valid since
\[
\sum_{j=2}^{T+1}2^j=2^{T+2}-4\le n-1.
\]
All background records lie outside the initial gate, so $z_t=q$ without a trigger. The movement bound excludes the remote point through time $T$.

Draw $n$ IID records from
\[
P_1=(1-1/n)P_0+(1/n)\delta_{x_*},
\]
and condition on one trigger. Let $N_j$ count background records at $x_j$, and set $N_1=1$. Under this conditional law, define
\[
E=\{N_3>1\}
 \cap\bigcap_{s=3}^T
 \{N_{s+1}>N_{s-1}+2N_{s-2}\},
\]
For $T=2$, omit the indexed intersection. Only the trigger is initially active, so $z_1=x_1$ and
\[
z_2-x_2
=\frac r2\frac{N_3-1}{1+N_2+N_3}e_1.
\]
On $E$, $z_2$ is strictly between $x_2,x_3$. If $z_{s-1}$ is strictly between $x_{s-1},x_s$, its gate contains only $x_{s-2},x_{s-1},x_s,x_{s+1}$, giving
\[
z_s-x_s
=\frac r2
\frac{N_{s+1}-N_{s-1}-2N_{s-2}}
{N_{s-2}+N_{s-1}+N_s+N_{s+1}}e_1.
\]
On $E$, the numerator is positive and $N_s>0$ by the preceding count inequality. Thus $z_s$ is strictly between $x_s,x_{s+1}$. Induction gives
\[
\|z_T-q\|>\frac{(T+1)r}{2}\ge\frac{Tr}{2}
\qquad\text{on }E.
\]

By the multinomial law,
\[
\mathbb E\prod_{j=2}^{T+1}u_j^{N_j}
=\left\{1+\frac1{n-1}
\sum_{j=2}^{T+1}2^j(u_j-1)\right\}^{n-1}.
\]
By Markov,
\[
\begin{aligned}
\mathbb P(N_3\le1)
&\le2\mathbb E2^{-N_3}
 =2\left(1-\frac4{n-1}\right)^{n-1}
 \le2e^{-4},\\
\mathbb P(N_4\le N_2+2)
&\le4\mathbb E2^{N_2-N_4}
 =4\left(1-\frac4{n-1}\right)^{n-1}
 \le4e^{-4},
\end{aligned}
\]
The second line requires $T\ge3$. For $s\ge4$, apply Markov to $(3/4)^{N_{s+1}-N_{s-1}-2N_{s-2}}$:
\[
\begin{aligned}
\mathbb P(N_{s+1}-N_{s-1}-2N_{s-2}\le0)
&\le\left(1-\frac{5\,2^{s-2}}{9(n-1)}\right)^{n-1}\\
&\le e^{-5\,2^{s-2}/9},
\end{aligned}
\]
because
\[
2^{s+1}(-1/4)+2^{s-1}(1/3)+2^{s-2}(7/9)
=-5\,2^{s-2}/9.
\]
The union bound, omitting absent terms for $T=2,3$, gives
\[
\begin{aligned}
\mathbb P(E^c)
&\le6e^{-4}+\sum_{k=0}^{\infty}e^{-(20/9)2^k}\\
&\le6e^{-4}+\sum_{k=0}^{\infty}e^{-(20/9)(k+1)}\\
&=6e^{-4}+\frac1{e^{20/9}-1}<\frac14.
\end{aligned}
\]
Thus $\mathbb P(\|z_T-q\|\ge Tr/2)\ge3/4$ under the conditional law.

\par\medskip\noindent\textbf{Private-mechanism lower bound.}
Take $M$ with $\Rterm(M)<\infty$. Draw $n$ IID background records to form $D_0$; replace a uniformly chosen row by $x_*$ to form $D_1$. Then $D_0\sim P_0^n$, $D_0\sim D_1$, and $D_1$ has the law of $P_1^n$ conditional on one trigger. That event has probability
\[
(1-1/n)^{n-1}\ge e^{-1}.
\]
Restrict this coupling to $\|z_T(D_1)-q\|\ge Tr/2$. Append the terminal-output kernel on $D_i$ to obtain measures $\mu_i$ on
\[
\Omega=(\mathbb R^d)^n\times(\mathbb R^d)^n\times\mathbb R^d,
\]
with common mass at least $3/4$. Integrating DP over sections gives
\[
\mu_1(A)-e^\varepsilon\mu_0(A)
\le\delta\mu_1(\Omega)
\qquad(A\subseteq\Omega\text{ measurable}).
\]
With $k_i=d\mu_i/d(\mu_0+\mu_1)$, define
\[
d\nu=\min\{k_1,e^\varepsilon k_0\}\,d(\mu_0+\mu_1).
\]
Taking positive parts gives
\[
\begin{aligned}
\nu(\Omega)
&=\mu_1(\Omega)
 -\int(k_1-e^\varepsilon k_0)_+\,d(\mu_0+\mu_1)\\
&\ge(1-\delta)\mu_1(\Omega)
 \ge\frac34(1-\delta).
\end{aligned}
\]
Also $\nu\le\mu_1$, $\nu\le e^\varepsilon\mu_0$, and on its support
\[
\frac{Tr}{2}\le\|y-q\|+\|y-z_T(D_1)\|.
\]
Since $z_T(D_0)=q$, the two risks are bounded by $\Rterm(M)$ and $e\Rterm(M)$. By Minkowski,
\[
\begin{aligned}
\frac{Tr}{2}\sqrt{\frac34(1-\delta)}
&\le\left(\int\|y-q\|^2\,d\nu\right)^{1/2}
 +\left(\int\|y-z_T(D_1)\|^2\,d\nu\right)^{1/2}\\
&\le\bigl(e^{\varepsilon/2}+\sqrt e\bigr)\sqrt{\Rterm(M)}.
\end{aligned}
\]
Square and take $\inf_M$:
\begin{equation}
\Rstar(\varepsilon,\delta)
\ge\frac{3(1-\delta)}{16(e^{\varepsilon/2}+\sqrt e)^2}T^2r^2.
\label{eq:app-long-lower}
\end{equation}
The bound also holds for infinite risk and for maximum-path risk.

\par\medskip\noindent\textbf{Conclusion.}
By DP, $M$ belongs to both minimax classes. Equations~\eqref{eq:app-output-projection-bound} and \eqref{eq:app-rare-lower} give
\[
\begin{aligned}
\frac{(1-\delta)e^{-2\varepsilon/3}}{32e}r^2
&\le\Rstar(\varepsilon,\delta)
 \le\Rterm(M)\le(2T+1)^2r^2,\\
\frac{(1-\delta)e^{-2\varepsilon/3}}{32e}r^2
&\le\mathcal R_{\max,T}^{\star}(\varepsilon,\delta)
 \le\mathcal R_{\max,T}(M)\le(2T+1)^2r^2.
\end{aligned}
\]
The constants are independent of $n,d,r$. In \eqref{eq:long-chain-region}, use \eqref{eq:app-long-lower} and $(2T+1)^2\le9T^2$:
\[
\begin{aligned}
\frac{3(1-\delta)T^2r^2}{16(e^{\varepsilon/2}+\sqrt e)^2}
&\le\Rstar(\varepsilon,\delta)
 \le\Rterm(M)\le9T^2r^2,\\
\frac{3(1-\delta)T^2r^2}{16(e^{\varepsilon/2}+\sqrt e)^2}
&\le\mathcal R_{\max,T}^{\star}(\varepsilon,\delta)
 \le\mathcal R_{\max,T}(M)\le9T^2r^2.
\end{aligned}
\]
These constants depend only on $\varepsilon,\delta$.
\end{proof}

\section{Proof of Theorem~\ref{thm:general-brace-first-order-limit}}
\label{app:clt-proofs}

\begin{proof}
By $\widetilde Y_{n,t}=Q_{n,t}+\xi_{n,t}$,
\[
\begin{aligned}
\sqrt n\,(\widetilde Y_{n,t}-z_{n,t})
&=-\sqrt n\,(z_{n,t}-a_{n,t}-H_{n,t}R_{n,t})\\
&\quad+\sqrt n\,\{Q_{n,t}-(a_{n,t}+H_{n,t}R_{n,t})\}
 +\sqrt n\,\xi_{n,t}.
\end{aligned}
\]
By Assumption~\ref{ass:brace-clt-policy}, fixed $T$, and Slutsky,
\[
\sqrt n\,(\widetilde Y_{n,1:T}-z_{n,1:T})
\Longrightarrow-C_{1:T}+\xi_{1:T}.
\]
Thus $\max_{t\le T}\|\widetilde Y_{n,t}-z_{n,t}\|=O_P(n^{-1/2})$.

If this maximum is less than $r$,
\[
\|\widetilde Y_{n,t}-q\|
\le\|\widetilde Y_{n,t}-z_{n,t}\|+\|z_{n,t}-q\|
<(t+1)r.
\]
Thus
\[
\begin{aligned}
\mathbb P\{Y_{n,1:T}\ne\widetilde Y_{n,1:T}\}
&\le\mathbb P\!\left\{
\max_{t\le T}\|\widetilde Y_{n,t}-z_{n,t}\|\ge r
\right\}\\
&\longrightarrow0.
\end{aligned}
\]
For every $\eta>0$,
\[
\begin{aligned}
\mathbb P\!\left(\sqrt n\,\|Y_{n,1:T}-\widetilde Y_{n,1:T}\|>\eta\right)
&\le\mathbb P(Y_{n,1:T}\ne\widetilde Y_{n,1:T})\longrightarrow0.
\end{aligned}
\]
Slutsky gives the same limit for $Y_{n,1:T}$.

Under joint Gaussianity, $-C_{1:T}+\xi_{1:T}$ is Gaussian with mean and covariance
\[
-\mathbb EC_{1:T},
\qquad
\operatorname{Cov}(-C_{1:T}+\xi_{1:T}),
\]
since $\mathbb E\xi_{1:T}=0$.
\end{proof}

\section{Proof of Theorem~\ref{thm:effective-count-effect}}
\label{app:effective-count-proof}

\begin{proof}
Write $Y_{n,<t}=(q,Y_{n,1},\ldots,Y_{n,t-1})$. Step~6 and $\overline\Delta_{n,t}=2r/\widehat m_{n,t}$ give
\[
\sqrt n\,\xi_{n,t}\mid(D_n,Y_{n,<t})
\sim\mathcal N\!\left(
0,\frac{2Tr^2}{\rho_{\varepsilon,\delta}}
\frac{n}{\widehat m_{n,t}^{\,2}}I_d
\right).
\]
For $u\in\mathbb R^d$,
\[
\mathbb E\!\left[e^{i\sqrt n\,u^\top\xi_{n,t}}
\mid D_n,Y_{n,<t}\right]
=\exp\!\left\{
-\frac{Tr^2}{\rho_{\varepsilon,\delta}}
\frac{n}{\widehat m_{n,t}^{\,2}}\|u\|^2
\right\}.
\]
If $\sqrt n/\widehat m_{n,t}\xrightarrow{P}\ell_t\in[0,\infty)$, this converges in probability to
\[
\exp\!\left\{
-\frac{Tr^2\ell_t^2}{\rho_{\varepsilon,\delta}}\|u\|^2
\right\}.
\]
The right side is bounded by one. Taking expectations and applying L\'evy's theorem gives
\[
\sqrt n\,\xi_{n,t}
\Longrightarrow\mathcal N\!\left(
0,\frac{2Tr^2\ell_t^2}{\rho_{\varepsilon,\delta}}I_d
\right).
\]
For $\ell_t=0$, $\sqrt n\,\xi_{n,t}\xrightarrow{P}0$.

Let $F_{\chi_d^2}$ be the $\chi_d^2$ distribution function. For $K>0$,
\[
\mathbb P\!\left(
\|\sqrt n\,\xi_{n,t}\|\le K\mid D_n,Y_{n,<t}
\right)
=F_{\chi_d^2}\!\left(
\frac{\rho_{\varepsilon,\delta}K^2}{2Tr^2}
\frac{\widehat m_{n,t}^{\,2}}{n}
\right).
\]
If $\widehat m_{n,t}/\sqrt n\xrightarrow{P}0$, continuity at zero and $0\le F_{\chi_d^2}\le1$ give
\[
\mathbb P\{\|\sqrt n\,\xi_{n,t}\|\le K\}\longrightarrow0.
\]
Thus $\|\sqrt n\,\xi_{n,t}\|\xrightarrow{P}\infty$, so the sequence is not tight.

If $\widehat m_{n,t}/n\xrightarrow{P}c_t>0$,
\[
\frac{n}{\widehat m_{n,t}}\xrightarrow{P}\frac1{c_t},
\qquad
\frac{\sqrt n}{\widehat m_{n,t}}\xrightarrow{P}0.
\]
Thus $\sqrt n\,\xi_{n,t}\xrightarrow{P}0$. The same argument with $n^2/\widehat m_{n,t}^{\,2}\xrightarrow{P}c_t^{-2}$ gives
\[
n\xi_{n,t}
\Longrightarrow\mathcal N\!\left(
0,\frac{2Tr^2}{\rho_{\varepsilon,\delta}c_t^2}I_d
\right).
\]
\end{proof}

\section{Proof of Corollary~\ref{cor:selected-brace-trajectory-limit}}
\label{app:selected-brace-proof}

\begin{proof}
For measurable $f$, write
\[
Pf=\mathbb Ef(X),\qquad
P_nf=\frac1n\sum_{i=1}^nf(X_i),\qquad
\mathbb G_nf=\sqrt n(P_n-P)f.
\]
Define
\[
g_y(x)=\mathbf1\{\|x-y\|\le r\},
\qquad
p(y)=Pg_y,
\qquad
\Phi_P(y)=\frac{P(Xg_y)}{p(y)}
\]
for $p(y)>0$. Then $p(\bar z_{t-1})=p_t$, $\Phi_P(\bar z_{t-1})=\bar z_t$, and $\|\bar z_t-q\|\le tr$ by induction.

\par\medskip\noindent
The continuous density is bounded near each compact boundary sphere. For small fixed $a>0$, set
\[
U_t=\overline B_2(\bar z_{t-1},a),
\qquad
U=\bigcup_{t=1}^TU_t.
\]
For $y,u\in U_t$, the gate difference lies in an annular shell of thickness $\|y-u\|$. Its volume and the density bound give
\[
P\!\left(\overline B_2(y,r)\mathbin{\triangle}
          \overline B_2(u,r)\right)
\le C\|y-u\|.
\]
All positive constants below are independent of $n$. For small enough $a$, $p_t>0$ gives
\begin{equation}
\begin{aligned}
\inf_{y\in U}p(y)&\ge2c>0,\\
|p(y)-p(u)|+\|\Phi_P(y)-\Phi_P(u)\|
&\le C\|y-u\|,\qquad y,u\in U_t.
\end{aligned}
\label{eq:app-local-regularity}
\end{equation}
For $\Phi_P$, use the gate-difference bound, positive denominators, and $\|X\|\le\|q\|+Tr+a$ inside the gates.

The classes $\{g_y:y\in U\}$ and $\{x_jg_y:y\in U\}$, $1\le j\le d$, are bounded VC-type classes. By \citet[Chapter~19]{vaart1998asymptotic},
\begin{equation}
\sup_{y\in U}
\bigl\{|\mathbb G_ng_y|+\|\mathbb G_n(Xg_y)\|\bigr\}
=O_P(1).
\label{eq:app-uniform-empirical-bound}
\end{equation}
Also,
\[
\begin{aligned}
P[(g_y-g_u)^2]&\le C\|y-u\|,\\
P[X_j^2(g_y-g_u)^2]&\le C\|y-u\|,
\qquad y,u\in U_t.
\end{aligned}
\]
By \eqref{eq:app-uniform-empirical-bound}, $\inf_{y\in U}P_ng_y\ge c$ with probability tending to one. On this event,
\[
\sqrt n\{\Phi_{D_n}(y)-\Phi_P(y)\}
=\frac{\mathbb G_n(Xg_y)-\Phi_P(y)\mathbb G_ng_y}{P_ng_y}.
\]
For $y_n,u_n\xrightarrow{P}\bar z_{t-1}$, stochastic equicontinuity and \eqref{eq:app-local-regularity} give
\begin{equation}
\begin{aligned}
\sqrt n\bigl[\Phi_{D_n}(y_n)-\Phi_P(y_n)
-\Phi_{D_n}(u_n)+\Phi_P(u_n)\bigr]
\xrightarrow{P}0.
\end{aligned}
\label{eq:app-empirical-equicontinuity}
\end{equation}

\par\medskip\noindent
For $n\ge2$, set
\[
k_n=\left\lceil\frac{2\log n}{\min_{1\le s\le T}\gamma_s}\right\rceil.
\]
Then $k_n=O(\log n)$ and eventually $k_n+1\le n$. Write $P_Bf=n^{-1}\sum_i f(B_i)$. For $d_H(D_n,B)\le k_n+1$, bounded gated functions give
\[
\begin{aligned}
\sup_{y\in U}|P_Bg_y-P_ng_y|
&\le\frac{k_n+1}{n},\\
\sup_{y\in U}\|P_B(Xg_y)-P_n(Xg_y)\|
&\le C\frac{k_n+1}{n}.
\end{aligned}
\]
By \eqref{eq:app-uniform-empirical-bound} and the positive denominators,
\begin{equation}
\sup_{d_H(D_n,B)\le k_n+1}\sup_{y\in U}
\|\Phi_B(y)-\Phi_P(y)\|
=O_P(n^{-1/2}+k_n/n)=O_P(n^{-1/2}).
\label{eq:app-robust-map-bound}
\end{equation}
Start from $z_0(B)=\bar z_0=q$. For $z_{s-1}(B)\in U_s$,
\[
\begin{aligned}
\|z_s(B)-\bar z_s\|
&\le\sup_{y\in U_s}\|\Phi_B(y)-\Phi_P(y)\|\\
&\quad+C\|z_{s-1}(B)-\bar z_{s-1}\|.
\end{aligned}
\]
Induction over fixed $T$, using \eqref{eq:app-robust-map-bound}, keeps the centers in $U_s$ with probability tending to one and gives
\begin{equation}
\sup_{d_H(D_n,B)\le k_n+1}\max_{0\le s\le T}
\|z_s(B)-\bar z_s\|=O_P(n^{-1/2}).
\label{eq:app-robust-path-bound}
\end{equation}

\par\medskip\noindent
Induct on $t$, starting from $Y_{n,0}=\bar z_0=q$. Suppose $Y_{n,s}-\bar z_s=O_P(n^{-1/2})$ for $s<t$. Using $R_s(B)=z_s(B)-\Phi_B(Y_{n,s-1})$ and \eqref{eq:app-local-regularity}--\eqref{eq:app-robust-path-bound},
\[
\sup_{d_H(D_n,B)\le k_n+1}\max_{1\le s\le t}
\|R_s(B;Y_{n,<s})\|=O_P(n^{-1/2}).
\]
Since $\kappa Tr^2>0$ is fixed, with probability tending to one,
\[
H_s(B;Y_{n,<s})=0
\qquad(d_H(D_n,B)\le k_n+1,\ 1\le s\le t).
\]
The uniform count bound also gives, with probability tending to one,
\[
|I(B,Y_{n,s-1})|\ge cn
\qquad(d_H(D_n,B)\le k_n+1,
\ 1\le s\le t).
\]
Every neighbor of a bank within distance $k_n$ lies within distance $k_n+1$. Step~5 then gives
\[
\sup_{d_H(D_n,B)\le k_n}
L_t(B;Y_{n,<t})\le\frac{2r}{cn}.
\]
For more distant banks, $L_t\le2r$. Splitting the envelope gives, with probability tending to one,
\[
\begin{aligned}
\overline\Delta_{n,t}
&\le\max\left\{
\frac{2r}{n},\frac{2r}{cn},2re^{-\gamma_t k_n}
\right\}\\
&\le\frac{C}{n},
\end{aligned}
\]
since $e^{-\gamma_t k_n}\le n^{-2}$. Thus $\overline\Delta_{n,t}=O_P(n^{-1})$.

By Step~6,
\[
\mathbb E[\|\xi_{n,t}\|^2\mid D_n,Y_{n,<t}]
=\frac{dT\overline\Delta_{n,t}^{\,2}}{2\rho_{\varepsilon,\delta}}.
\]
For $K,b>0$, conditional Markov gives
\[
\mathbb P(n\|\xi_{n,t}\|>b)
\le\mathbb P(n\overline\Delta_{n,t}>K)
+\frac{dTK^2}{2\rho_{\varepsilon,\delta}b^2}.
\]
Thus $\xi_{n,t}=O_P(n^{-1})$. With probability tending to one, $H_{n,t}=0$ and $Q_{n,t}=\Phi_{D_n}(Y_{n,t-1})$. By \eqref{eq:app-local-regularity} and \eqref{eq:app-robust-map-bound},
\[
\begin{aligned}
\|Q_{n,t}-\bar z_t\|
&\le\|\Phi_{D_n}(Y_{n,t-1})-\Phi_P(Y_{n,t-1})\|\\
&\quad+C\|Y_{n,t-1}-\bar z_{t-1}\|
=O_P(n^{-1/2}).
\end{aligned}
\]
Since $\bar z_t$ is in the output ball, nonexpansiveness gives
\[
\|Y_{n,t}-\bar z_t\|
\le\|Q_{n,t}+\xi_{n,t}-\bar z_t\|
=O_P(n^{-1/2}).
\]
Induction and fixed $T$ give
\begin{equation}
\begin{aligned}
\mathbb P\{H_{n,t}=0\text{ for all }t\le T\}&\longrightarrow1,\\
\max_{t\le T}\overline\Delta_{n,t}&=O_P(n^{-1}),\\
\max_{t\le T}\|\xi_{n,t}\|&=O_P(n^{-1}).
\end{aligned}
\label{eq:app-selected-noise-scale}
\end{equation}

\par\medskip\noindent
Induct from $Y_{n,0}=z_{n,0}=q$. Both preceding centers converge to $\bar z_{t-1}$. Equation~\eqref{eq:app-empirical-equicontinuity} gives
\[
\begin{aligned}
\sqrt n\,(a_{n,t}-z_{n,t})
&=\sqrt n\bigl\{
\Phi_P(Y_{n,t-1})-\Phi_P(z_{n,t-1})
\bigr\}+o_P(1).
\end{aligned}
\]
By local Lipschitz continuity and induction, this is $o_P(1)$, so $\sqrt n\,R_{n,t}\xrightarrow{P}0$. Equation~\eqref{eq:app-selected-noise-scale} gives $Q_{n,t}=a_{n,t}$ with probability tending to one and $\sqrt n\,\xi_{n,t}\xrightarrow{P}0$. By nonexpansiveness,
\[
\sqrt n\,\|Y_{n,t}-z_{n,t}\|
\le\sqrt n\,\|Q_{n,t}-z_{n,t}\|
+\sqrt n\,\|\xi_{n,t}\|
\xrightarrow{P}0.
\]

Also,
\[
\sqrt n\,(z_{n,t}-a_{n,t}-H_{n,t}R_{n,t})
=\sqrt n\,(1-H_{n,t})R_{n,t}\xrightarrow{P}0.
\]
By \eqref{eq:app-selected-noise-scale},
\[
\begin{aligned}
&\mathbb P\!\left(\sqrt n\max_{t\le T}
\|Q_{n,t}-(a_{n,t}+H_{n,t}R_{n,t})\|>\eta\right)\\
&\qquad\le\mathbb P\{H_{n,t}=1\text{ for some }t\le T\}\longrightarrow0,
\qquad\eta>0.
\end{aligned}
\]
Together with $\sqrt n\max_{t\le T}\|\xi_{n,t}\|\xrightarrow{P}0$, this proves Assumption~\ref{ass:brace-clt-policy} with $C_t=\xi_t=0$. Joint convergence follows from fixed $T$.
\end{proof}

\subsection{MNIST reproducibility details}

\paragraph{Code and data.}
The supplementary archive \texttt{brace\_mnist\_reproducibility.zip} contains the code, package versions, reproduction commands, VAE checkpoint, PCA features, data splits, corrupted queries, Gaussian draws, and reported results. The training script downloads MNIST and checks the files.

\paragraph{Training and settings.}
Pixel values are divided by 255. With seed 20271006, we use 50,000 training images, a disjoint 1,000-image memory bank, 128 development images, and 256 test queries. The VAE uses a 784--400 encoder, ten-dimensional mean and log-variance heads, and a 10--400--784 decoder. Training uses Adam with learning rate $10^{-3}$, batch size 256, and 15 epochs. The loss is binary cross-entropy plus KL divergence, with KL weight $\min(1,e/3)$ at epoch $e$. PCA is fitted to the training posterior means without whitening.

Retrieval uses seed 202610063, $q_i=x_i+3u_i$ with $u_i\sim\mathcal N(0,I_2)$, and 32 noise repetitions. The same starts and standard Gaussian draws are used across methods. We use
\[
\beta\in\{10^{j/5}:j=-30,\ldots,15\}\cup\{12.7427,20.6914\},
\qquad
T\in\{3,5,7\},\quad
\varepsilon\in\{16,32,64\},
\]
with $\delta=10^{-5}$, $r=\sqrt{2/\beta}$, $\kappa=128$, and $\gamma=0.25$.

\paragraph{Metrics and confidence intervals.}
For $Q=256$, $R=32$, and $d=2$,
\[
\widehat E(a)=\frac{1}{QRd}
\sum_{i=1}^{Q}\sum_{j=1}^{R}
\|Y_T^{(i,j)}-a_i\|^2.
\]
Prediction MSE uses $a_i=x_i$, while fidelity MSE uses $a_i=z_{i,T}$. Pointwise 95\% intervals use 512 query-cluster bootstrap resamples. Repetitions are averaged within each query, queries are resampled with replacement, and the 2.5th and 97.5th percentiles are reported.

\end{document}